\documentclass[11pt]{article}
\usepackage[margin=1in]{geometry}
\usepackage{graphicx}
\usepackage{amsmath}
\usepackage{amsthm}
\usepackage{physics}
\usepackage{amssymb}
\usepackage{xcolor}
\usepackage{hyperref}
\usepackage{mathtools}
\usepackage{comment}
\usepackage{subcaption}
\usepackage[normalem]{ulem}
\usepackage{mathrsfs}

\theoremstyle{plain}
\newtheorem{theorem}{Theorem}[section]
\newtheorem{lemma}[theorem]{Lemma}
\newtheorem{proposition}[theorem]{Proposition}
\newtheorem{corollary}[theorem]{Corollary}

\theoremstyle{definition}

\theoremstyle{remark}

\newtheorem{claim}[theorem]{Claim}
\newtheorem{remark}[theorem]{Remark}

\numberwithin{equation}{section}

\usepackage{dsfont}
\newcommand{\eps}{\epsilon}
\newcommand{\indicator}[1]{\mathds{1}_{#1}}
\newcommand{\set}[1]{\left\{#1\right\}}
\newcommand{\ceil}[1]{\left\lceil#1\right\rceil}
\newcommand{\floor}[1]{\left\lfloor#1\right\rfloor}
\DeclarePairedDelimiter{\parentheses}{(}{)}
\renewcommand{\P}{\mathbb{P}}
\newcommand{\R}{\mathbb{R}}
\newcommand{\V}{\mathbb{V}}
\newcommand{\Z}{\mathbb{Z}}
\newcommand{\cC}{\mathcal{C}}
\newcommand{\cE}{\mathcal{E}}
\newcommand{\cF}{\mathcal{F}}
\newcommand{\cG}{\mathcal{G}}
\newcommand{\cL}{\mathcal{L}}
\newcommand{\cP}{\mathcal{P}}
\newcommand{\fA}{\mathfrak{A}}
\newcommand{\fe}{\mathfrak{e}}
\newcommand{\fm}{\mathfrak{m}}
\newcommand{\fn}{\mathfrak{n}}
\newcommand{\fs}{\mathfrak{s}}
\newcommand{\sF}{\mathscr{F}}
\newcommand{\oned}{\mathrm{1D}}
\newcommand{\hor}{\mathrm{hor}}
\newcommand{\per}{\mathrm{per}}
\newcommand{\ver}{\mathrm{ver}}
\newcommand{\mic}{\mathrm{mic}}
\newcommand{\Cov}{\mathrm{Cov}}
\newcommand{\eigenval}{x}
\newcommand{\chempot}{\lambda}
\newcommand\rect[2]{\mathrm{R}_{#1,#2}}
\newcommand\peierls[1]{\mathfrak{p}_{#1}}
\newcommand\es[1]{\overline{#1}}
\newcommand\xre[2]{X_{#1}#2}
\newcommand{\Area}[1]{\mathrm{Area}(#1)}
\newcommand{\Width}[1]{\mathrm{Width}(#1)}
\newcommand{\Height}[1]{\mathrm{Height}(#1)}
\newcommand{\Perimeter}[1]{\mathrm{Perimeter}(#1)}
\newcommand{\Phase}[1]{\mathrm{Phase}(#1)}

\DeclarePairedDelimiter{\brackets}{[}{]}

\let\emptyset\varnothing

\allowdisplaybreaks

\title{Extended regime of nematic order in an interacting monomer-dimer model of Heilmann and Lieb}
\author{Qidong He}
\date{}

\begin{document}

\maketitle

\begin{abstract}
    We revisit a two-dimensional model of liquid crystals introduced by Heilmann--Lieb (1979), which consists of dimers on the square lattice at chemical potential $\chempot$, interacting via a hard-core repulsion and an attractive interaction of strength $a>0$ between adjacent, colinear dimers.
    The model is conjectured to exhibit nematic order at low temperatures, in the sense of orientational symmetry breaking coupled with the absence of translational order, provided that $\chempot+a>0$.
    We prove this conjecture under the additional condition $3a>\chempot$, substantially extending the parameter regime in which nematic order is known.
    The extra condition corresponds exactly to the regime in which misaligned dimers remain rare on scales comparable to the correlation length of the one-dimensional reference system.
    Our proof begins with the microscopic characterization of orientational order used by Heilmann--Lieb, strengthens it into a quantitative characterization of the rarity of misaligned dimers under the strong-percolation formalism of Hadas--Peled (2025), and uses the improved estimate to implement a disagreement-percolation argument in the sense of van den Berg (1993).
    As a consequence, we establish uniqueness of Gibbs measure and anisotropic exponential decay of correlations within each orientational phase.
\end{abstract}

\section{Introduction}

Liquid crystal is a state of matter characterized microscopically by the presence of orientational order with partial or no positional order, resulting in properties intermediate between those of a liquid and a solid crystal~\cite{de1993physics}.
Mathematical models of liquid crystals generally fall into one of two complementary categories: macroscopic theories based on continuum mechanics, and molecular theories rooted in statistical mechanics~\cite{wang2021modelling}.
While macroscopic theories~\cite{oseen1933theory,frank1958liquid} often treat the origin of liquid-crystalline phases phenomenologically, molecular theories aim to derive these phases directly from microscopic interactions.
Central to the latter approach are various toy models of hard-core particles which, following Onsager's seminal work on three-dimensional hard rods~\cite{onsager1949effects}, demonstrate that entropic effects alone can drive the transition from a disordered, isotropic phase to a liquid-crystalline phase that favors global geometric alignment.

Recent simulations of lattice systems of hard-core particles have revealed a rich landscape of liquid-crystalline phases dictated strictly by particle geometry.
In two dimensions, systems of $1\times k$ hard rods on the square lattice with discrete (horizontal or vertical) orientations are found~\cite{kundu2013nematic} to exhibit a \emph{nematic} phase at intermediate densities for $k\ge 7$ and proven so~\cite{disertori2013nematic} in the limit $k\gg1$.
Systems of $2\times 2$ hard squares on the square lattice have been shown~\cite{fernandes2007monte,nath2014multiple} and very recently proven~\cite{hadas2025columnar} to exhibit a \emph{columnar} phase at high densities.
In three dimensions, systems of $2\times 2\times 2$ hard cubes on the cubic lattice have been found via Monte Carlo simulations~\cite{vigneshwar2019phase} to exhibit four distinct phases, including a \emph{layered} phase at intermediate densities and a columnar phase at high densities.
Similar findings for other hard-core shapes and on different lattices have also been reported~\cite{thewes2020phase,mazel2025high,mandal2023phases,vigneshwar2017different}; see~\cite[Section 10]{hadas2025columnar} for a recent survey.
Despite the evident phenomenological richness, rigorous results in this area remain sparse compared to standard lattice spin models, largely due to difficulties in implementing rigorous perturbation methods such as Pirogov--Sinai theory~\cite{zahradnik1984alternate,cannon2026pirogov} for these systems, of which~\cite{disertori2013nematic} is a singular exception.

In this paper, we study a two-dimensional lattice model for liquid crystals introduced by Heilmann--Lieb in 1979~\cite[Model I]{heilmann1979lattice}.
In this model, particles are represented by dimers (edges) on the square lattice $\Z^2$, and the interactions between particles consist of a hard-core repulsion and an attractive interaction between adjacent, colinear dimers.
With $\chempot\in\R$ denoting the dimer chemical potential and $a>0$ denoting the strength of the attractive interaction, Heilmann--Lieb proved, using reflection-positivity arguments, that the model exhibits orientational symmetry breaking, in the sense of the existence of two Gibbs measures that respectively favor vertical and horizontal dimers, at low temperatures provided that $\chempot+a>0$.
However, in view of the two defining qualities of a liquid crystal, they stopped short of proving the absence of translational symmetry breaking within each phase.
Regarding this omission, Jauslin--Lieb~\cite{jauslin2018nematic} wrote decades later that 
\begin{quote}
    ``The main difficulty with extending [Heilmann--Lieb's] method to prove the lack of translational order is that the correlation length of the system is very large\dots and the lack of order is only visible on that scale, and seems difficult to see using only reflection positivity.''
\end{quote}

Over the years, there have been two main mathematical works aimed specifically at completing the proof of liquid-crystalline order in this model.
The first attempt is due to Alberici in 2016~\cite{alberici2016cluster}, who used the cluster expansion to prove the absence of translational symmetry breaking under the assumption of unequal horizontal and vertical dimer activities.
Shortly afterwards, in 2018, Jauslin--Lieb~\cite{jauslin2018nematic}, using Pirogov--Sinai theory with a sophisticated contour representation of the model, resolved the problem in the original setting of equal horizontal and vertical dimer activities.
Notwithstanding their success, the fly in the ointment is the \emph{extreme} parameter regime in which they worked: whereas it is natural to compare $\chempot$ and $a$ \emph{directly} from a physical point of view as in the original work of Heilmann--Lieb, their computation required an exceptionally strong attractive interaction in the sense that $\beta a\gg e^{\beta\chempot}\gg 1$. 
It has remained an open problem to bridge the vast gap between the regime considered by Heilmann--Lieb and that by Jauslin--Lieb.

This paper takes a step towards closing this gap.
Specifically, we prove using reflection-positivity methods that the model exhibits nematic order at low temperatures in the parameter regime
\begin{equation}
    \chempot+a>0,\quad 3a>\chempot,
\end{equation}
noting that the constraint $a>0$ is built into the model.
As in~\cite{heilmann1979lattice}, the condition $\chempot+a>0$ ensures that vacancies are suppressed at low temperatures, leading to densely packed configurations from which orientational order is readily visible from the dimers themselves.
The added condition $3a>\chempot$ plays a different role: it ensures that isolated misaligned dimers, each associated with a penalty factor of $e^{-3\beta a}$, remain rare on scales comparable to the correlation length $\sim e^{\frac{1}{2}\beta(\chempot+3a)}$ of the one-dimensional reference system.
This separation of scales is precisely what allows us to construct \emph{sealed rectangles} in the disagreement-percolation argument to control the propagation of disagreement components.
We will reflect on this point in Section~\ref{sec:origin of extra condition}.
In the meantime, whether nematic order persists throughout the entire conjectured regime $\chempot+a>0$ remains open.

The observation of Jauslin--Lieb quoted above accurately identifies the main difficulty addressed in this paper.
Reflection-positivity readily yields orientational order through microscopic considerations, but it does not automatically rule out the existence of translational order on a mesoscopic scale.
Our argument accesses the latter in two steps.
First, starting from the microscopic characterization of orientational order used by Heilmann--Lieb~\cite{heilmann1979lattice}, we prove a near-optimal quantitative characterization of the rarity of misaligned dimers (Theorem~\ref{thm: microscopic rarity of horizontal dimers}) under the strong-percolation formalism of Hadas--Peled~\cite[Section~6]{hadas2025columnar}.
Second, we combine this with one-dimensional transfer-matrix estimates to prove the strong percolation of \emph{sealed rectangles} of width one and length comparable to $\ell_0$ (Proposition~\ref{prop:sealed rectangles strongly percolate}).
Within such a rectangle, disagreement paths, in the sense of van den Berg~\cite{van1993uniqueness}, are confined to a single column and terminated by coincident vacancies or dimers near its ends (Proposition~\ref{prop:bounded disagreement components in sealed rectangles}). 
The strong percolation of these rectangles thus enables the disagreement-percolation method to convert the microscopic information supplied by reflection-positivity into uniqueness (Proposition~\ref{prop:unique ergodic Gibbs measure for each phase}) and exponential decay of correlations (Proposition~\ref{prop:disagreement percolation decay}) within each orientational phase.

Hence, our proof builds on the strong-percolation and disagreement-percolation framework developed by Hadas--Peled, but exploits an important simplification specific to the present model.
Namely, since a dimer has an intrinsic orientation (horizontal or vertical), the microscopic characterization of orientational order considered by Heilmann--Lieb already suffices to distinguish the two orientational phases.
This stands in contrast to the hard-square model studied by Hadas--Peled, where orientational order is only discernible on a mesoscopic scale~\cite[Section~4]{hadas2025columnar} because a square has no intrinsic orientation.
Hence, this geometric property of dimers makes it possible, though with some effort, to pass directly from the microscopic to the mesoscopic.\footnote{We thank an anonymous referee for suggesting and outlining this route; an earlier version of this paper instead proceeded through a direct adaptation of Hadas--Peled~\cite{hadas2025columnar}, where the proof started on the mesoscopic scale and was substantially more involved.}
We note that the condition $3a>\chempot$ enters crucially into this passage, and refer the reader to Section~\ref{sec:origin of extra condition} for a detailed discussion as well as some initial thoughts on the complementary regime $3a\le\chempot$.

\paragraph{Model and main results}

Without further ado, we now introduce the model and state the main results.

Let $\V$ be the edge set of the square lattice $\Z^2$.
We identify each element of $\V$ with its midpoint in $\R^2$, so that $\V\equiv[(1/2+\Z)\times\Z]\cup[\Z\times(1/2+\Z)]$.
A \textbf{dimer configuration} is a function $\sigma:\V\to\set{0,1}$ such that $\sigma(e)\sigma(e')=0$ whenever $e,e'\in\V$ are distinct and incident to the same vertex of $\Z^2$.
We identify each dimer configuration $\sigma$ with the set $\sigma^{-1}(\set{1})$ of non-adjacent edges of $\Z^2$.
Let $\Omega$ be the set of all dimer configurations on $\V$.

Let $\beta>0$ be the inverse temperature, $\chempot\in \R$ be the chemical potential, and $a>0$ be the interaction strength.
The formal Hamiltonian of the monomer-dimer model, as Heilmann--Lieb introduced it in~\cite{heilmann1979lattice}, is
\begin{equation}
    \label{eqn:formal Hamiltonian}
    H(\sigma)= -\chempot\abs{\sigma}-a\sum_{\set{e,e'}\subset\sigma}\indicator{e\sim e'},
\end{equation}
where $\abs{\cdot}$ denotes the cardinality, and $e\sim e'$ means that $e,e'\in\V$ are colinear and separated by exactly one other edge of $\Z^2$.
The formal Hamiltonian~\eqref{eqn:formal Hamiltonian} is adapted to finite volumes in the usual way and gives rise to the usual finite-volume Gibbs measures as well as infinite-volume ones via the standard Dobrushin--Lanford--Ruelle formalism.
A detailed review of the definitions and constructions referenced above is given in Section~\ref{sec:the Hamiltonian}.

Our main result is as follows.

\begin{theorem}[Nematic order]
    \label{thm:main}
    Fix $\chempot\in\R$ and $a>0$ such that $\chempot+a>0$ and $3a>\chempot$.
    There exist constants $\beta_0,C,c>0$ and $\fn\ge 100$ (see~\eqref{eqn: microscopic length scales} for an exact expression) such that, for all $\beta\ge\beta_0$, the monomer-dimer model admits a unique Gibbs measure $\mu_\ver$ such that
    \begin{enumerate}
        \item \label{itm:main - invariance and extremality} $\mu_\ver$ is $\Z^2$-invariant and extremal.
        \item \label{itm:main - dimer probabilities} For all $e\in\V$,
        \begin{equation}
            \mu_\ver(\sigma(e)=1)=
            \begin{cases}
                \frac{1}{2}-\order{e^{-\frac{1}{2}\beta (\chempot+a)}} & \text{if }e\text{ is vertical} \\ 
                \order{e^{-3\beta(1-1/\fn) a}} & \text{otherwise}
            \end{cases}.
        \end{equation}
        \item \label{itm:main - decay of correlations} Let $A,B\subset\V$ and $f,g:\Omega\to[-1,1]$ be such that $f$ depends only on the restriction of $\sigma$ to $A$ and $g$ only on the restriction of $\sigma$ to $B$.
        Then,
        \begin{equation}
            \abs{\Cov_{\mu_\ver}(f,g)}\le \sum_{u\in A}\sup_{v\in B}\alpha_1(u,v),
        \end{equation}
        where, defining 
        $\ell_0:=e^{\frac{1}{2}\beta(\chempot+3a)}$,
        \begin{equation}
            \alpha_1((x_u,y_u),(x_v,y_v)):=Ce^{-c\abs{x_v-x_u}-c\ell_0^{-1}\abs{y_v-y_u}}.
        \end{equation}
    \end{enumerate}
    Let $\tau:\R^2\to\R^2$ be the reflection $(x,y)\mapsto(y,x)$.
    The push-forward of $\mu_\ver$ under $\tau$, denoted by $\mu_\hor$, is another Gibbs measure of the monomer-dimer model that satisfies analogous properties to the above.
    Furthermore, every periodic Gibbs measure of the model is a convex combination of $\mu_\ver$ and $\mu_\hor$.
\end{theorem}

\paragraph{Organization of the paper}

In Section~\ref{sec:preliminaries}, we formally define the monomer-dimer model, review the reflection-positivity property and the chessboard estimate, and import the formalism of strongly percolating sets from~\cite{hadas2025columnar}.
In Section~\ref{sec:one-dimensional systems}, we analyze the one-dimensional monomer-dimer model by the transfer-matrix method and derive the basic estimates used later in the proof. 
In Section~\ref{sec: characterizations of orientational symmetry breaking}, we first formulate the Heilmann--Lieb characterization of the orientational phases in the strong-percolation framework, and then strengthen it to a near-optimal quantitative characterization of the rarity of misaligned dimers in each phase.
In Section~\ref{sec: absence of translational symmetry breaking}, we construct sealed rectangles, use the previous estimates to prove their strong percolation on the mesoscopic scale, implement disagreement percolation to prove uniqueness of Gibbs measure and anisotropic exponential decay of correlations within each phase, and characterize the two phases microscopically by estimating the densities of vertical and horizontal dimers.
Finally, in Section~\ref{sec:origin of extra condition}, we reflect on the precise role of the condition $3a>\chempot$ in our argument and discuss the obstacles to extending the proof to the regimes $3a=\chempot$ and $3a<\chempot$.

\section{Preliminaries}
\label{sec:preliminaries}

\subsection{Basic definitions}
\label{sec:basic definitions}

\paragraph{Graphical constructions}

In this paper, we will consider three graphical constructions based on the square lattice $\Z^2$:
\begin{enumerate}
    \item The \textbf{dual square lattice}, denoted by $(\Z^2)^\ast$, has vertices corresponding to the faces of $\Z^2$ (with respect to its standard embedding in $\R^2$), and two vertices are connected by an edge if and only if their corresponding faces share an edge.
    We will identify $(\Z^2)^\ast$ with the shifted square lattice $(1/2,1/2)+\Z^2$.
    \item The \textbf{line graph} of $\Z^2$, denoted by $L(\Z^2)$, has vertices corresponding to the edges of $\Z^2$, and two vertices are connected by an edge if and only if their corresponding edges in $\Z^2$ share a vertex.
    We will identify the vertices of $L(\Z^2)$ with the midpoints of the edges of $\Z^2$; thus, $L(\Z^2)$ has vertex set $\V\equiv[(1/2+\Z)\times\Z]\cup[\Z\times(1/2+\Z)]$.
    \item An \emph{augmentation} of the line graph of $\Z^2$, denoted by $L^\ddag(\Z^2)$, is obtained from $L(\Z^2)$ by adding an edge between each pair of vertices of $L(\Z^2)$ corresponding to colinear edges of $\Z^2$ that are separated by exactly one other edge.
    We will continue to identify the vertices of $L^\ddag(\Z^2)$ with the midpoints of the edges of $\Z^2$ as before, and use the symbol $\ddag$ to denote the connectivity on $L^\ddag(\Z^2)$.
\end{enumerate}

Thus, each $e\in\V$ may be viewed both as an edge of $\Z^2$ and as a vertex of $L(\Z^2)$ and $L^\ddag(\Z^2)$.
Unless otherwise noted, we adopt the former viewpoint from here on, so the terms ``vertex'' and ``edge'' refer to a vertex and an edge of $\Z^2$, respectively.

Our motivation for introducing the \emph{augmented line graph} $L^\ddag(\Z^2)$ is that the monomer-dimer model is a Markov random field~\cite[(1.6)]{van1994disagreement} on $\V$ with respect to $\ddag$-connectivity---due to the nature of the interactions in the model (see~\eqref{eqn:broken-link potential}), the same is not true with respect to the usual connectivity on $L(\Z^2)$.

Furthermore, we will consider two distinct connectivities on $\Z^2$ itself.
The first, denoted by the symbol $\Box$, is the standard connectivity on $\Z^2$ induced by the $1$-norm on $\R^2$ that makes $\Z^2$ into the square lattice graph.
The second, denoted by the symbol $\boxtimes$, is induced by the sup-norm on $\R^2$ and is particularly useful for characterizing the connectivity of boundaries of $\Box$-connected subsets of $\Z^2$~\cite{timar2013boundary}.

\paragraph{Rectangles}

We use the word \textbf{rectangle} to refer to a closed, axis-parallel rectangle in $\R^2$ with corners on the dual square lattice: for $x,y\in 1/2+\Z$ and $K,L\in\Z_{>0}$, denote
\begin{equation}
    \rect{K\times L}{(x,y)}:=[x,x+K]\times[y,y+L]\subset\R^2.
\end{equation}
We also introduce the shorthand $\mathrm{R}_{K\times L}:=\rect{K\times L}{(-1/2,-1/2)}$.
We say that a rectangle is even if both its width and height are even.

\paragraph{Dimer configurations}

Recall that we have identified $\V=[(1/2+\Z)\times\Z]\cup[\Z\times(1/2+\Z)]$ with the edge set of $\Z^2$.
A \textbf{dimer} is an edge $e\in\V$.
A \textbf{dimer configuration} on $\V$ is a function $\sigma:\V\rightarrow\set{0,1}$ with the property that no two edges in $\sigma^{-1}(\set{1})$ are incident.
By identifying the function $\sigma$ with the set $\sigma^{-1}(\set{1})$, a dimer configuration becomes simply a matching of $\Z^2$.
Denote the set of all dimer configurations on $\V$ by $\Omega$.
Equip $\Omega$ with the subspace topology inherited from the product topology on $\set{0,1}^\V$ and with the corresponding Borel measurable structure; denote the resulting $\sigma$-algebra by $\sF$.

Let $\Lambda\subset\R^2$ be a rectangle.
We introduce several restricted sets of dimer configurations corresponding to various boundary conditions on $\Lambda$:
\begin{enumerate}
    \item \emph{Periodic boundary conditions}.
    Define the set of $\Lambda$-periodic dimer configurations as
    \begin{equation}
        \Omega^\per_\Lambda:=\set{\sigma\in\Omega\mid\text{for all }e\in\V,\sigma(e)=\sigma(e+(\Width{\Lambda},0))=\sigma(e+(0,\Height{\Lambda}))}.
    \end{equation}
    \item \emph{Prescribed boundary conditions}.
    Let $\rho\in\Omega$.
    Define the set of dimer configurations in $\Lambda$ (more precisely, in $\V\cap\Lambda$) with $\rho$-boundary conditions as
    \begin{equation}
        \Omega^\rho_\Lambda:=\set{\sigma\in\Omega\mid\text{for all }e\in\V\setminus\Lambda,\sigma(e)=\rho(e)}.
    \end{equation}
    \item \emph{Vacant boundary conditions}.
    Define the set of dimer configurations in $\Lambda$ with vacant boundary conditions as
    \begin{equation}
        \Omega^0_\Lambda:=\set{\sigma\in\Omega\mid\text{for all }v\in\Z^2\setminus\Lambda\text{, }v\text{ is a vacancy in }\sigma}.
    \end{equation}
\end{enumerate}

Finally, given a set $A\subset\V$ and a measurable function $f:\Omega\to\R$, we say that $f$ is \textbf{$A$-local} if the value of $f$ depends only on the restriction of a dimer configuration to $A$.

\begin{remark}
    As a cautionary note, the way that the monomer-dimer model is defined on the edge set of $\Z^2$ may be a source of confusion for readers accustomed to lattice spin models where configurations are specified directly on the vertices.
    Our notation $\V$ for the edge set of $\Z^2$ is intended to mitigate this confusion.
\end{remark}

\subsection{The Hamiltonian and Gibbs measures}
\label{sec:the Hamiltonian}

Let $\sigma\in\Omega$.
Recall that we have identified $\sigma\equiv\sigma^{-1}(\set{1})$.
We say that two dimers $e,e'\in\sigma$ are \textbf{linked} if they are colinear and separated by exactly one other edge called their \textbf{link}; in that case, we write $e\sim e'$.
Let $\beta>0$ be the inverse temperature, $\chempot\in\R$ be the dimer chemical potential, and $a>0$ be the strength of attraction between linked dimers.
In the remainder of the paper, we will consider $\chempot$ and $a$ to be fixed such that $\chempot+a>0$ and $3a>\chempot$, and suppress them from the notation.
The formal Hamiltonian of the monomer-dimer model, as introduced by Heilmann--Lieb in~\cite[Model I]{heilmann1979lattice}, is
\begin{equation}
    H(\sigma):=-\chempot\abs{\sigma}-a\sum_{\set{e,e'}\subset\sigma}\indicator{e\sim e'},
\end{equation}
where $\abs{\sigma}$ is the number of dimers in the dimer configuration $\sigma$.

In this paper, we will study their model in the perturbative regime $\beta\gg1$.
To exploit the largeness of $\beta$, it is advantageous to shift the Hamiltonian by a constant multiple of the volume, as exemplified in the following formal rewriting, whose adaptation to finite volumes is straightforward (though it may lead to an inconsequential boundary term):
\begin{equation}
    \label{eqn:shifted Hamiltonian monomer-dimer}
    \frac{1}{2}(\chempot+a)\abs{\Z^2}+H(\sigma)
    =\frac{1}{2}(\chempot+a)(\abs{\Z^2}-2\abs{\sigma})+\frac{1}{2}a\sum_{e\in\sigma}\left(2-\sum_{e'\in\sigma}\indicator{e'\sim e}\right),
\end{equation}
where the terms on the RHS have the following intuitive interpretations:
\begin{enumerate}
    \item $\abs{\Z^2}-2\abs{\sigma}$ is the number of vacant vertices, or \textbf{vacancies} for short, which are vertices of $\Z^2$ not incident to any dimer in $\sigma$.
    \item For each $e\in\sigma$, $2-\sum_{e'\in\sigma}\indicator{e'\sim e}$ is the number of ends of $e$ not linked to another dimer in $\sigma$.
    Formally, we define a \textbf{broken link} as an edge in $\V$ which is incident to exactly one dimer in $\sigma$ with which it is colinear.
    Hence, $\sum_{e\in\sigma}\left(2-\sum_{e'\in\sigma}\indicator{e'\sim e}\right)$ is the total number of broken links in $\sigma$.
\end{enumerate}

For clarity, we recast~\eqref{eqn:shifted Hamiltonian monomer-dimer} systematically using the standard notion of potentials (see~\cite[Section 6.3.2]{friedli2017statistical} or~\cite[(2.2)]{georgii2011gibbs}).
Let $B\subset\V$ be bounded (or $B\Subset\V$ for short).
\begin{enumerate}
    \item If $B$ has the form $\set{e_1,e_2,e_3,e_4}$ where all the $e_i$ are incident to the same vertex of $\Z^2$, define the four-body, vacancy potential
    \begin{equation}
        \label{eqn:vacancy potential}
        \Phi_B(\sigma):=\frac{1}{2}(\chempot+a)\indicator{\sum_{i=1}^4\sigma(e_i)=0},
    \end{equation}
    which outputs $(\chempot+a)/2$ if and only if the shared vertex is a vacancy in $\sigma$.
    \item If $B$ has the form $\set{e_1,e_2,e_3}$ where the $e_i$ (ordered from left to right or bottom to top) are adjacent and colinear, define the three-body, broken-link potential
    \begin{equation}
        \label{eqn:broken-link potential}
        \Phi_B(\sigma):=\frac{1}{2}a\indicator{\sigma(e_1)+\sigma(e_3)=1},
    \end{equation}
    which outputs $a/2$ if and only if the middle edge $e_2$ is a broken link of the dimer $e_1$ or $e_3$ in $\sigma$.
    \item Otherwise, set $\Phi_B(\sigma):=0$.
\end{enumerate}
Thus, the shifted formal Hamiltonian can be expressed as
\begin{equation}
    \sum_{B\Subset\V}\Phi_B(\sigma),
\end{equation}
whose adaptation to finite volumes is standard (see~\cite[(6.24)]{friedli2017statistical} or~\cite[(2.3)]{georgii2011gibbs}): if $\Lambda$ is a rectangle and $\#$ denotes a boundary condition for $\Lambda$, then the finite-volume Hamiltonian for dimer configurations in $\Lambda$ with $\#$-boundary conditions is
\begin{equation}
    \label{eqn:finite-volume Hamiltonian}
    H^\#_{\Lambda}(\sigma):=\sum_{\substack{B\Subset\V\\ B\cap\Lambda\ne\emptyset}}\Phi_B(\sigma).
\end{equation}
As an exception, we note that if $\#$ denotes the periodic boundary conditions, then the sum over $B\Subset\V$ is instead restricted to a complete set of representatives of the equivalence classes of all $B\Subset\V$ under translations generated by $(\Width{\Lambda},0)$ and $(0,\Height{\Lambda})$.

\paragraph{Gibbs measures}

Let $\Lambda$ be a rectangle and $\#$ denote a boundary condition for $\Lambda$.
The weight of a dimer configuration $\sigma\in\Omega^\#_\Lambda$ is defined as
\begin{equation}
    w^\#_{\Lambda;\beta}(\sigma):=e^{-\beta H^\#_{\Lambda}(\sigma)}.
\end{equation}
Hence, each vacancy carries a penalty factor of $e^{-\frac{1}{2}\beta(\chempot+a)}$, whereas each broken link carries a penalty factor of $e^{-\frac{1}{2}\beta a}$.
The finite-volume Gibbs measure on $\Lambda$ with $\#$-boundary conditions is given by
\begin{equation}
    \mu^\#_{\Lambda;\beta}(\sigma):=\frac{w^\#_{\Lambda;\beta}(\sigma)}{Z^\#_{\Lambda;\beta}}\quad\text{for }\sigma\in\Omega^\#_\Lambda,
\end{equation}
where $Z^\#_{\Lambda;\beta}:=\sum_{\sigma\in\Omega^\#_\Lambda}w^\#_{\Lambda;\beta}(\sigma)$ is the partition function.
Furthermore, given $E\subset\Omega$, we define the $E$-constrained partition function as
\begin{equation}
    Z^\#_{\Lambda;\beta}(E):=\sum_{\sigma\in\Omega^\#_\Lambda\cap E}w^\#_{\Lambda;\beta}(\sigma).
\end{equation}
Lastly, we define infinite-volume Gibbs measures via the standard Dobrushin--Lanford--Ruelle formalism; see~\cite[Section 6.2.1]{friedli2017statistical} or~\cite[(2.9)]{georgii2011gibbs} for details.

\subsection{Reflection positivity and the chessboard estimate}

In this subsection, we formulate the reflection-positivity of the monomer-dimer model and use it to set up the chessboard estimate and an infinite-volume extension thereof due to Hadas--Peled~\cite[Proposition 3.8]{hadas2025columnar}.
While much of the material here is either standard or directly imported from~\cite{hadas2025columnar}, we will give a further extension of the chessboard estimate applicable to \emph{finite products} of periodic Gibbs measures that will be useful in Section~\ref{sec:disagreement percolation}.

\subsubsection{Standard results}

We begin with some notation.
Let $R=\rect{K\times L}{(x_0,y_0)}$ be a rectangle.
Define the \textbf{grid} of $R$ as $G^R:=(x_0+K\Z)\times(y_0+L\Z)$, and denote its origin-shifted version by $\cL^R:=K\Z\times L\Z$.
Let $T^R$ be the group generated by the reflections of $R$ across the vertical and horizontal lines that pass through points in $G^R$.
Thus, for each $v\in G^R$, there exists a unique isometry in $T^R$ that maps $R$ to $\rect{K\times L}{v}$; we denote this isometry by $\tau_{R,v}$.
Given $f:\Omega\to\R$ and $\tau\in T^R$, define $(\tau f)(\sigma):=f(\sigma\circ \tau)$.
Given a Gibbs measure $\mu$ and an isometry $\tau:\R^2\to\R^2$, denote by $\tau\mu$ the push-forward of $\mu$ under $\tau$.
Given $(x,y)\in\Z^2$, let $\eta_{(x,y)}:\R^2\to\R^2$ denote the translation $v\mapsto v+(x,y)$.

First, we show that the monomer-dimer model is \emph{reflection-positive} with respect to reflections across the edges in $\V$.

\begin{lemma}[Reflection positivity]
    Let $R=\rect{K\times L}{(x_0,y_0)}$ be a rectangle and $f:\Omega\to\R$ be $R$-local.
    Then,
    \begin{equation}
        \mu_{\Lambda;\beta}^\per(f\cdot\tau f)\ge 0
    \end{equation}
    if $(\Lambda,\tau)$ is either $(\rect{2K\times L}{(-1/2,-1/2)},\tau_{R,(x_0+K,y_0)})$ or $(\rect{K\times 2L}{(-1/2,-1/2)},\tau_{R,(x_0,y_0+L)})$.
\end{lemma}

\begin{proof}
    For the purpose of this proof, it is convenient to temporarily separate the hard-core constraint from the definition of dimer configurations by enlarging the configuration space to $\widetilde{\Omega}:=\set{0,1}^\V$, and enforcing the hard-core constraint through the following modification of the four-body potential~\eqref{eqn:vacancy potential}:
    \begin{equation}
        \label{eqn: modified vacancy potential}
        \widetilde{\Phi}_B(\sigma):=
        \begin{cases}
            \frac{1}{2}(\chempot+a) & \text{if }\sum_{e\in B}\sigma(e)=0\\
            0 & \text{if }\sum_{e\in B}\sigma(e)=1\\
            \infty & \text{otherwise}
        \end{cases}
        ,
    \end{equation}
    which is applied to the finite-volume Hamiltonian~\eqref{eqn:finite-volume Hamiltonian} in the natural way.

    We treat the case that $(\Lambda,\tau)=(\rect{2K\times L}{(-1/2,-1/2)},\tau_{R,(x_0+K,y_0)})$; the other case is similar.
    Let
    \begin{equation}
        \fA^+:=\set{u+v\mid u\in \cL^\Lambda,\, v\in R},
        \quad
        \fA^-:=\set{u+v\mid u\in \cL^\Lambda,\, v\in\rect{K\times L}{(x_0+K,y_0)}}.
    \end{equation}
    By a standard criterion for reflection-positivity (see~\cite[Corollary 5.4]{biskup2009reflection} or~\cite[Lemma 10.8]{friedli2017statistical}), it suffices to write 
    \begin{equation}
        -\widetilde{H}_\Lambda^\per=A+\tau A+\sum_{\alpha}C_\alpha\cdot\tau C_\alpha,
    \end{equation}
    where $A,C_\alpha$ are $\fA^+$-local functions.
    Recall the potential representation of the Hamiltonian~\eqref{eqn:finite-volume Hamiltonian}.
    We let each $\widetilde{\Phi}_B$ as in~\eqref{eqn: modified vacancy potential} contribute to $-(A+\tau A)$.
    For each $\Phi_B$ as in~\eqref{eqn:broken-link potential}, if $B\subset\fA^+$ or $B\subset\fA^-$, we let it contribute to $-(A+\tau A)$.
    Otherwise, $B$ crosses a horizontal or vertical line passing through points in $G^R$.
    In this case, writing
    \begin{equation}
        \Phi_B(\sigma)=\frac{1}{2}a-\frac{1}{2}a\left[\indicator{\sigma(e_1)=\sigma(e_2)}\indicator{\sigma(e_3)=\sigma(e_2)}+\indicator{\sigma(e_1)\ne\sigma(e_2)}\indicator{\sigma(e_3)\ne\sigma(e_2)}\right], 
    \end{equation}
    we let the first term contribute to $-(A+\tau A)$ and the second term to $-\sum_\alpha C_\alpha\cdot\tau C_\alpha$, using that $a>0$.
    The proof is now complete.
\end{proof}

If $\Lambda$ is a rectangle such that $2\Width{R}\mid\Width{\Lambda}$ and $2\Height{R}\mid\Height{\Lambda}$, we say that $R$ is a \textbf{block} of $\Lambda$.
In that case, set $T^R_\Lambda:=T^R/\cL^\Lambda$.
Given an $R$-local function $f:\Omega\to\R$, define its \textbf{$(R,\Lambda)$-chessboard seminorm} by 
\begin{equation}
    \norm{f}_{R\mid\Lambda}:=\mu_{\Lambda;\beta}^\per\parentheses[\Bigg]{\prod_{\tau\in T_\Lambda^R}\tau f}^{1/\abs{T_\Lambda^R}};
\end{equation}
the name ``seminorm'' is justified by the fact that $\norm{\cdot}_{R\mid\Lambda}$ indeed satisfies the standard properties of a seminorm~\cite[Proposition 3.3]{hadas2025columnar}.
A standard consequence of reflection-positivity is the \emph{chessboard estimate}, which asserts that the expectation of a product of local functions of a particular form under $\mu_{\Lambda;\beta}^\per$ factorizes as an upper bound into the product of their chessboard seminorms.
We do not prove the chessboard estimate here but refer the reader to~\cite{biskup2009reflection} and~\cite{friedli2017statistical} for pedagogical references.

\begin{proposition}[Chessboard estimate]
    \label{prop:chessboard estimate}
    Let $R$ be a block of $\Lambda$.
    Let $A\subset T_\Lambda^R$ and $(f_\tau)_{\tau\in A}$ be $R$-local functions.
    Then,
    \begin{equation}
        \mu_{\Lambda;\beta}^\per\parentheses[\Bigg]{\prod_{\tau\in A}\tau f_\tau}
        \le\prod_{\tau\in A}\norm{f_\tau}_{R\mid\Lambda}.
    \end{equation}
\end{proposition}

We also have the following recursive version of the chessboard estimate, taken from~\cite{hadas2025columnar}.

\begin{lemma}[Recursive chessboard estimate, {\cite[Lemma 3.4]{hadas2025columnar}}]
    \label{lem:recursive chessboard estimate}
    Let $R$ and $S$ be blocks of $\Lambda$, and suppose that the corners of $S$ are in $G^R$.
    Let $A\subset T^R$ be such that $\bigcup_{\tau\in A}\tau R\subset S$ and $(f_\tau)_{\tau\in A}$ be $R$-local functions.
    Then,
    \begin{equation}
        \norm{\prod_{\tau\in A}\tau f_\tau}_{S\mid\Lambda}\le\prod_{\tau\in A}\norm{f_\tau}_{R\mid\Lambda}.
    \end{equation}
\end{lemma}

\begin{proof}
    The assumptions imply that each $(\kappa,\tau)\in T_\Lambda^S\times A$ gives rise to a unique element $\kappa\tau\in T_\Lambda^R$.
    By Proposition~\ref{prop:chessboard estimate}, 
    \begin{equation}
        \norm{\prod_{\tau\in A}\tau f_\tau}_{S\mid\Lambda}^{\abs{T_\Lambda^S}}
        =\mu_{\Lambda;\beta}^\per\parentheses[\Bigg]{\prod_{\kappa\in T_\Lambda^S}\kappa\parentheses[\Bigg]{\prod_{\tau\in A}\tau f_\tau}}
        =\mu_{\Lambda;\beta}^\per\parentheses[\Bigg]{\prod_{\kappa\in T_\Lambda^S}\prod_{\tau\in A}\kappa\tau f_\tau}
        \le\prod_{\tau\in A}\norm{f_\tau}_{R\mid\Lambda}^{\abs{T_\Lambda^S}},
    \end{equation}
    which completes the proof.
\end{proof}

\subsubsection{Finite products of reflection-positive measures}
\label{sec:infinite-volume chessboard estimate for product measures}

The goal of this sub-subsection is to adapt the \textbf{infinite-volume chessboard estimate} of Hadas--Peled~\cite[Proposition 3.8]{hadas2025columnar} to the monomer-dimer model, and, more generally, to finite products of its periodic Gibbs measures. 
Here, we say that a Gibbs measure is \textbf{periodic} if it is invariant under translations by a full-rank sublattice of $\Z^2$.

Let $k\ge 1$.
An element $(\sigma_i)_{i=1}^k\in\Omega^k$ may be thought of as a joint configuration that naturally specifies a mapping $\V\to\set{0,1}^k$ via $e\mapsto(\sigma_i(e))_{i=1}^k$.
Thus, given a rectangle $R$, the notion of $R$-locality and transformation under each reflection $\tau\in T^R$ extend naturally to functions $\Omega^k\rightarrow\R$.
We equip $\Omega^k$ with the product $\sigma$-algebra $\sF^{\otimes k}$, and use the symbol $\otimes$ to denote the product of probability measures on $(\Omega,\sF)$.

Let $\Lambda$ be a rectangle, $R$ be a block of $\Lambda$, and $f:\Omega^k\to\R$ be $R$-local.
Define
\begin{equation}
    \norm{f}_{R\mid\Lambda}:=(\otimes_{i=1}^k \mu_{\Lambda;\beta}^\per)\parentheses[\Bigg]{\prod_{\tau\in T_\Lambda^R}\tau f}^{1/\abs{T_\Lambda^R}}.
\end{equation}
We note that the product measure $\otimes_{i=1}^k\mu_{\Lambda;\beta}^\per$ remains reflection-positive with respect to reflections across the edges in $\V$, so the quantity above satisfies the same seminorm properties as before and the product-measure analogues of Proposition~\ref{prop:chessboard estimate} and Lemma~\ref{lem:recursive chessboard estimate} hold.
Next, define the \textbf{$R$-chessboard seminorm} of $f$ by
\begin{equation}
    \label{eqn:infinite-volume chessboard seminorm}
    \norm{f}_R:=\limsup_{n\to\infty}\norm{f}_{R\mid\mathrm{R}_{n!\times n!}}.
\end{equation}
The resulting product-measure analogue of~\cite[Proposition 3.8]{hadas2025columnar} is as follows.

\begin{proposition}[Infinite-volume chessboard estimate for product measures]
    \label{prop:infinite-volume chessboard estimate for product measures}
    Let $R$ be a rectangle. 
    Let $A\subset T^R$ be finite and $(f_\tau)_{\tau\in A}$ be $R$-local functions from $\Omega^k$ to $\R$.
    Then, given periodic Gibbs measures $(\mu_i)_{i=1}^k$,
    \begin{equation}
        (\otimes_{i=1}^k \mu_i)\parentheses[\Bigg]{\prod_{\tau\in A}\tau f_\tau}\le\prod_{\tau\in A}\norm{f_\tau}_R.
    \end{equation}
\end{proposition}

The proof of Proposition~\ref{prop:infinite-volume chessboard estimate for product measures} follows the same strategy as in~\cite{hadas2025columnar}, with suitable modifications to accommodate the product measures.
The proof of the following lemma is technical and is deferred to Section~\ref{sec: proof of technical lemma}.

\begin{lemma}
    \label{lem:infinite-volume chessboard estimate for product measures - single function}
    Let $R$ be a rectangle and $f:\Omega^k\to\R$ be $R$-local.
    Then, given periodic Gibbs measures $(\mu_i)_{i=1}^k$,
    \begin{equation}
        (\otimes_{i=1}^k \mu_i)(f)\le\norm{f}_R.
    \end{equation}
\end{lemma}

We also need the following extension of Lemma~\ref{lem:recursive chessboard estimate}, which follows directly from~\eqref{eqn:infinite-volume chessboard seminorm}.

\begin{lemma}
    \label{lem:infinite-volume recursive chessboard estimate for product measures}
    Let $R$ and $S$ be rectangles, and suppose that the corners of $S$ are in $G^R$.
    Let $A\subset T^R$ be such that $\bigcup_{\tau\in A}\tau R\subset S$ and $(f_\tau)_{\tau\in A}$ be $R$-local functions.
    Then,
    \begin{equation}
        \norm{\prod_{\tau\in A}\tau f_\tau}_{S}\le\prod_{\tau\in A}\norm{f_\tau}_{R}.
    \end{equation}
\end{lemma}

We now deduce Proposition~\ref{prop:infinite-volume chessboard estimate for product measures}.

\begin{proof}[Proof of Proposition~\ref{prop:infinite-volume chessboard estimate for product measures}]
    Let $S$ be a rectangle whose corners lie on $G^R$ and which contains $\bigcup_{\tau\in A}\tau R$.
    Combining Lemmas~\ref{lem:infinite-volume chessboard estimate for product measures - single function} and~\ref{lem:infinite-volume recursive chessboard estimate for product measures}, we get that
    \begin{equation}
        (\otimes_{i=1}^k \mu_i)\parentheses[\Bigg]{\prod_{\tau\in A}\tau f_\tau}
        \le\norm{\prod_{\tau\in A}\tau f_\tau}_S
        \le\prod_{\tau\in A}\norm{f_\tau}_R,
    \end{equation}
    as required.
\end{proof}

\subsubsection{Proof of Lemma~\ref{lem:infinite-volume chessboard estimate for product measures - single function}}
\label{sec: proof of technical lemma}

The proof of Lemma~\ref{lem:infinite-volume chessboard estimate for product measures - single function} relies on the following claims.

\begin{claim}[{\cite[Proposition 2.2]{hadas2025columnar}}]
    \label{clm:comparison of boundary conditions - expectation}
    There exists a constant $C(\beta)\ge 1$ such that the following holds.
    Let $R_0\subset R\subset\Lambda$ be rectangles such that the Euclidean distance from $R_0$ to $\R^2\setminus R$ is at least $2$. 
    Let $f:\Omega\to\R_{\ge 0}$ be $R_0$-local.
    Then, for any Gibbs measure $\mu$,
    \begin{equation}
        \mu(f)\le C(\beta)^{\Perimeter{R}}\mu^\per_{\Lambda;\beta}(f).
    \end{equation}
\end{claim}

The formulation of Claim~\ref{clm:comparison of boundary conditions - expectation} differs slightly from its source but is convenient for our purposes, since the application of its extension (Claim~\ref{clm:infinite-volume chessboard estimate for product measures - change of measures}) later on involves taking the limits $\Lambda\Uparrow\R^2$ and $R\Uparrow\R^2$ separately.
Its proof is based on standard corridor arguments and applies to the monomer-dimer model with minimal modification, so we refer the reader to~\cite{hadas2025columnar} for details.

\begin{claim}
    \label{clm:infinite-volume chessboard estimate for product measures - change of measures}
    Let $R_0\subset R\subset\Lambda$ be rectangles such that the Euclidean distance from $R_0$ to $\R^2\setminus R$ is at least $2$. 
    Let $f:\Omega^k\to\R_{\ge 0}$ be $R_0$-local.
    Let $(\mu_i)_{i=1}^k$ be periodic Gibbs measures.
    Then, with the same constant $C(\beta)\ge 1$ as in Claim~\ref{clm:comparison of boundary conditions - expectation},
    \begin{equation}
        (\otimes_{i=1}^k \mu_i)(f)\le C(\beta)^{k\Perimeter{R}}(\otimes_{i=1}^k \mu_{\Lambda;\beta}^\per)(f).
    \end{equation}
\end{claim}

\begin{proof}
    If $f$ is the indicator function of an $R_0$-local product event $E_1\times\dots\times E_k\subset\Omega^k$, the product measure factorizes into a product of measures, and the inequality follows from Claim~\ref{clm:comparison of boundary conditions - expectation}.
    Next, if $f$ is the indicator function of a general $R_0$-local event $B\subset\Omega^k$, the inequality follows from the construction of product measures via Carath\'eodory's extension theorem, i.e., the identity
    \begin{equation}
        \begin{multlined}
            (\otimes_{i=1}^k \mu_i)(B)\\
            =\inf\Bigg\{\sum_{j=1}^\infty(\otimes_{i=1}^k \mu_i)(B_j)\mid B_j\text{ is a finite union of $R_0$-local product events},\,B\subset\bigcup_{j=1}^\infty B_j\Bigg\}
        \end{multlined}
    \end{equation}
    and an analogue for $(\otimes_{i=1}^k \mu_{\Lambda;\beta}^\per)(B)$.
    By linearity, the inequality extends to non-negative $R_0$-local simple functions, and finally to all non-negative $R_0$-local functions by the monotone convergence theorem.
\end{proof}

\begin{claim}
    \label{clm:infinite-volume chessboard estimate for product measures - limsup bound}
    Let $(\mu_i)_{i=1}^k$ be periodic Gibbs measures, $R$ be a rectangle, and $g:\Omega^k\to\R_{\ge 0}$ be $R$-local. 
    Let $(A_n)_{n\ge 1}$ be non-empty, finite subsets of $T^R$.
    Suppose that $\abs{A_n}\to\infty$ as $n\to\infty$, and, denoting by $R_n$ the smallest rectangle containing $\bigcup_{\tau\in A_n}\tau R$, that $\Perimeter{R_n}/\abs{A_n}\to 0$ as $n\to\infty$.
    Then,
    \begin{equation}
        \limsup_{n\to\infty}\sqrt[\abs{A_n}]{(\otimes_{i=1}^k \mu_i)\parentheses[\Bigg]{\prod_{\tau\in A_n}\tau g}}
        \le\norm{g}_R.
    \end{equation}
\end{claim}

\begin{proof}
    Fix $n\ge 1$.
    For $m\ge 2$, define $\Lambda_m:=\rect{m!\times m!}{(-(m!+1)/2,-(m!+1)/2)}$.
    Let $R_n'\supset R_n$ be the smallest rectangle such that the Euclidean distance from $R_n$ to $\R^2\setminus R_n'$ is at least $2$, and $m(n)\ge 2$ be such that $R_n'\subset\Lambda_{m(n)}$.
    Suppose that $m\ge m(n)$.
    Since $\prod_{\tau\in A_n}\tau g$ is $R_n$-local, using Claim~\ref{clm:infinite-volume chessboard estimate for product measures - change of measures}, we get that
    \begin{equation}
        (\otimes_{i=1}^k \mu_i)\parentheses[\Bigg]{\prod_{\tau\in A_n}\tau g}
        \le C(\beta)^{k\Perimeter{R_n'}}(\otimes_{i=1}^k \mu_{\Lambda_m;\beta}^\per)\parentheses[\Bigg]{\prod_{\tau\in A_n}\tau g}.
    \end{equation} 
    When $m$ is so large that $R$ is a block of $\Lambda_m$, by the usual chessboard estimate applied to $\otimes_{i=1}^k \mu_{\Lambda_m;\beta}^\per$, we have that
    \begin{equation}
        (\otimes_{i=1}^k \mu_{\Lambda_m;\beta}^\per)\parentheses[\Bigg]{\prod_{\tau\in A_n}\tau g}
        \le\norm{g}_{R\mid\Lambda_m}^{\abs{A_n}}.
    \end{equation}
    Hence,
    \begin{equation}
        \sqrt[\abs{A_n}]{(\otimes_{i=1}^k \mu_i)\parentheses[\Bigg]{\prod_{\tau\in A_n}\tau g}}
        \le C(\beta)^{k\Perimeter{R_n'}/\abs{A_n}}\norm{g}_{R\mid\Lambda_m}.
    \end{equation}
    Using that $\Perimeter{R_n'}=\Perimeter{R_n}+16$, taking the limit superior of both sides as $m\to\infty$ and then as $n\to\infty$ completes the proof.
\end{proof}

\begin{claim}
    \label{clm:infinite-volume chessboard estimate for product measures - Taylor expansion}
    Let $M\in\Z_{>0}$, $S_n\subset M\Z^2$ be finite, and $A_n:=\set{\eta_s\mid s\in S_n}$ where $\eta_s:\R^2\to\R^2$ is the translation $u\mapsto u+s$.
    Let $0<\eps<1/2$ and $g:\Omega^k\to[1-\eps,1+\eps]$ be measurable.
    If $\otimes_{i=1}^k \mu_i$ is $M\Z^2$-invariant, then
    \begin{equation}
        (\otimes_{i=1}^k \mu_i)\parentheses[\Bigg]{\sqrt[\abs{A_n}]{\prod_{\tau\in A_n}\tau g}}
        =(\otimes_{i=1}^k \mu_i)(g)+\order{\eps^2}.
    \end{equation}
\end{claim}

The proof of Claim~\ref{clm:infinite-volume chessboard estimate for product measures - Taylor expansion} is by a direct application of the Taylor expansion and does not depend on the specifics of the model or the measure, so we omit it here and refer the reader to~\cite[Proof of Claim 3.11]{hadas2025columnar} for details.

We are now ready to prove Lemma~\ref{lem:infinite-volume chessboard estimate for product measures - single function}.

\begin{proof}[Proof of Lemma~\ref{lem:infinite-volume chessboard estimate for product measures - single function}]
    Let $(\mu_i)_{i=1}^k$ be periodic Gibbs measures.
    By taking the intersection of full-rank sublattices with respect to which each $\mu_i$ is invariant, we find a full-rank sublattice $\cL\subset\Z^2$ such that $\otimes_{i=1}^k \mu_i$ is $\cL$-invariant.
    Let $M\in\Z_{>0}$ be such that $M\Z^2\subset\cL\cap2\cL^R$.
    For each $n\in\Z_{>0}$, let $S_n:=\set{M,2M,\dots,nM}^2$ and $A_n$ be as in Claim~\ref{clm:infinite-volume chessboard estimate for product measures - Taylor expansion}. 
    For $0<\eps<\frac{1}{2\max\abs{f}}$, let $g:=1+\eps f$.
    Then,
    \begin{equation}
        \begin{multlined}
            1+\eps\norm{f}_R
            \ge\norm{g}_R
            \ge\limsup_{n\to\infty}\sqrt[n^2]{(\otimes_{i=1}^k \mu_i)\parentheses[\Bigg]{\prod_{\tau\in A_n}\tau g}}
            \ge\limsup_{n\to\infty}(\otimes_{i=1}^k \mu_i)\parentheses[\Bigg]{\sqrt[n^2]{\prod_{\tau\in A_n}\tau g}}
            \\
            \ge 1+\eps(\otimes_{i=1}^k \mu_i)(f)+\order{\eps^2},
        \end{multlined}
    \end{equation}
    where we used Claim~\ref{clm:infinite-volume chessboard estimate for product measures - limsup bound} in the second, Jensen's inequality in the third inequality, and Claim~\ref{clm:infinite-volume chessboard estimate for product measures - Taylor expansion} in the last.
    Taking $\eps\to0$ completes the proof.
\end{proof}

\subsection{Strongly percolating sets}
\label{sec:strongly percolating sets}

A recurrent theme in our analysis is to take a local property of dimer configurations---the existence of long dimer chains through the bulk of a fixed-size rectangle, for instance---and quantify its prevalence in typical dimer configurations via percolation-type arguments.
In this section, we import the notion of strongly percolating sets and their properties from~\cite[Section 6]{hadas2025columnar}, which are well-suited for this purpose.

We start with some notation.
Let $E\subset\Omega$ be an event, which should be thought of as a local property of dimer configurations.
Define the random set $\es{E}:\Omega\to 2^{\Z^2}$ by 
\begin{equation}
    \es{E}(\sigma):=\set{(x,y)\in\Z^2\mid\sigma\in\eta_{(x,y)}E},
\end{equation}
where we recall that $\eta_{(x,y)}$ denotes the translation by $(x,y)$.
Thus, $\es{E}$ is the random set of locations at which the local property $E$ holds.
Given $K,L\in\Z_{>0}$, define the random set
\begin{equation}
    \xre{K\times L}{\es{E}}:=\set{(x,y)\in\Z^2\mid(Kx,Ly)\in\es{E}},
\end{equation}
which tracks the manifestation of the local property $E$ on the scale of the rectangular grid $K\Z\times L\Z$.
The above definitions generalize naturally to \emph{joint} properties of multiple, independently sampled dimer configurations, i.e., events $E\subset\Omega^k$, $k\ge 2$.

We now define strongly percolating sets.
Fix $\eps_0:=1/21$ throughout the paper.
Let $\eps\ge 0$.
A random set $B\subset\Z^2$ on a probability space $(\Omega,\sF,\P)$ is \textbf{$\eps$-rare} if, for all finite $A\subset\Z^2$, $\P(A\subset B)\le\eps^{\abs{A}}$.
A random set $\Psi\subset\Z^2$ on $(\Omega,\sF,\P)$ is \textbf{$\eps$-strongly percolating} if either $\eps\ge\eps_0$ or there exists an $\eps$-rare set $B$ such that $\Psi$ $\P$-almost surely contains an infinite $\Box$-component of $\Z^2\setminus B$.
Define
\begin{equation}
    \label{eqn: strong percolation parameter}
    \peierls{\P}(\Psi):=\inf\set{\eps\ge 0\mid\Psi\text{ is }\eps\text{-strongly percolating}}.
\end{equation}
We mention here that the quantity $\peierls{\P}(\Psi)$ may be thought of as a measure of the size of $\Psi$, with smaller values indicating larger random sets; see~\cite[Section 6.1]{hadas2025columnar} for further discussion.

The following results are direct imports from~\cite[Section 6]{hadas2025columnar}.

\begin{lemma}[{\cite[Lemma 6.2]{hadas2025columnar}}]
    \label{lem:Peierls argument}
    If $\eps<\eps_0$ and $B$ is $\eps$-rare with respect to a probability measure $\P$, then $\Z^2\setminus B$ $\P$-almost surely contains a unique infinite $\Box$-component $I$.
    Moreover, each $\boxtimes$-component of $\Z^2\setminus I$ is finite.
\end{lemma}

\begin{lemma}[{\cite[Corollary 6.3]{hadas2025columnar}}]
    \label{lem:strong percolation of complement of rare set}
    If $\eps\ge 0$ and $B$ is $\eps$-rare with respect to a probability measure $\P$, then $\Z^2\setminus B$ is $\eps$-strongly percolating, and in particular $\peierls{\P}(\Z^2\setminus B)\le\eps$.
\end{lemma}

\begin{lemma}[{\cite[Lemma 6.4]{hadas2025columnar}}]
    \label{lem:quantitative Peierls argument}
    Let $u\in\Z^2$, $d\ge 0$, and $\Psi\subset\Z^2$ be a random set on a probability space $(\Omega,\sF,\P)$.
    Let $E$ be the event that there exists a $\boxtimes$-path in $\Z^2\setminus\Psi$ starting at $u$ and ending at some point in $\set{v\in\Z^2\mid\norm{v-u}_\infty\ge d}$.
    Then, 
    \begin{equation}
        \P(E)\le\left(\frac{\peierls{\P}(\Psi)}{\eps_0}\right)^{d+1}.
    \end{equation}
\end{lemma}

\begin{lemma}[{\cite[Lemma 6.6]{hadas2025columnar}}]
    \label{lem:intersection of strongly percolating sets}
    For all $k\ge 1$ and random sets $\Psi_1,\dots,\Psi_k\subset\Z^2$ on a common probability space $(\Omega,\sF,\P)$,
    \begin{equation}
        \peierls{\P}\parentheses[\Bigg]{\bigcap_{i=1}^k\Psi_i}\le k\sqrt[k]{\max_{1\le i\le k}\peierls{\P}(\Psi_i)}.
    \end{equation}
\end{lemma}

\begin{lemma}[{\cite[Lemma 6.7]{hadas2025columnar}}]
    \label{lem:strongly percolating grid events}
    Let $E\subset\Omega$ be an event.
    Let $k,K,l,L\in\Z_{>0}$ be such that $k\mid K$ and $l\mid L$.
    Denote $H:=\set{\eta_{(x,y)}\mid (x,y)\in(k\Z\times l\Z)\cap([0,K)\times[0,L))}$ and $r:=\abs{H}=\frac{KL}{kl}$.
    Denote $\Psi:=\xre{k\times l}{\es{E}}$, $\Psi'_\eta:=\xre{K\times L}{\es{\eta E}}$ for each $\eta\in H$, and $\Psi':=\bigcap_{\eta\in H}\Psi'_\eta$.
    Then,
    \begin{enumerate}
        \item \label{itm:strongly percolating grid events - subgrid over grid} $\peierls{\P}(\Psi')\le r\peierls{\P}(\Psi)$.
        \item $\peierls{\P}(\Psi)\le \sqrt[r]{\peierls{\P}(\Psi')}$.
        \item \label{itm:strongly percolating grid events - grid over subgrid}  $\peierls{\P}(\Psi)\le \sqrt[r]{r\sqrt[r]{\max_{\eta\in H}\peierls{\P}(\Psi'_\eta)}}$.
        \item \label{itm:strongly percolating grid events - chessboard bound} If $E\subset\Omega$ is $R$-local for $R=\rect{K\times L}{(x,y)}$ and is invariant under reflections through the vertical and horizontal lines passing through the center of $R$, then, for all periodic Gibbs measures $\mu$,
        \begin{equation}
            \peierls{\mu}(\Psi)\le\sqrt[r]{\norm{\Omega\setminus E}_R}.
        \end{equation}
    \end{enumerate}
\end{lemma}

\begin{lemma}[{\cite[Proposition 6.8]{hadas2025columnar}}]
    \label{lem:splitting strongly percolating sets}
    Let $\P$ be an $\cL$-ergodic probability measure on $\Omega$ for some lattice $\cL\subset K\Z\times L\Z$.
    Let $E,F\subset\Omega$ be events.
    If
    \begin{equation}
        E\cap F=E\cap\eta_{(K,0)}F=E\cap\eta_{(0,L)}F=\eta_{(K,0)}E\cap F=\eta_{(0,L)}E\cap F=\emptyset,\footnote{We thank an anonymous referee for pointing out that the original statement of~\cite[Proposition 6.8]{hadas2025columnar} incorrectly omitted the assumption that $E\cap F=\emptyset$.}
    \end{equation}
    then 
    \begin{equation}
        \min\set{\peierls{\P}(\xre{K\times L}{\es{E}}),\peierls{\P}(\xre{K\times L}{\es{F}})}\le\peierls{\P}(\xre{K\times L}{\es{E\cup F}}).
    \end{equation}
\end{lemma}

We conclude this subsection with the following extension of Item~\ref{itm:strongly percolating grid events - chessboard bound} of Lemma~\ref{lem:strongly percolating grid events}.

\begin{lemma}
    \label{lem:strongly percolating grid events - chessboard bound for product measures}
    Item~\ref{itm:strongly percolating grid events - chessboard bound} of Lemma~\ref{lem:strongly percolating grid events} applies to finite products of periodic Gibbs measures.
    Namely, under the setting of Lemma~\ref{lem:strongly percolating grid events}, if $E\subset\Omega^n$ is $R$-local for $R=\rect{K\times L}{(x,y)}$ and is invariant under reflections through the vertical and horizontal lines passing through the center of $R$, then, given periodic Gibbs measures $(\mu_i)_{i=1}^n$,
        \begin{equation}
            \peierls{\otimes_{i=1}^n\mu_i}(\Psi)\le \sqrt[r]{\norm{\Omega^n\setminus E}_R}.
        \end{equation}
\end{lemma}

\begin{proof}
    Let $\eps:=\norm{\Omega^n\setminus E}_R$.
    By Lemma~\ref{lem:strong percolation of complement of rare set}, it suffices to show that the random set $B:=\xre{k\times l}{\es{E^c}}$ is $\sqrt[r]{\eps}$-rare with respect to $\otimes_{i=1}^n\mu_i$.
    Let $A\subset\Z^2$ be finite.
    By the pigeonhole principle, there exists $\eta=\eta_{(x_0,y_0)}\in H$ such that, with $G:=(\frac{K}{k}\Z+\frac{x_0}{k})\times(\frac{L}{l}\Z+\frac{y_0}{l})$, $\abs{A\cap G}\ge\abs{A}/r$.
    By Proposition~\ref{prop:infinite-volume chessboard estimate for product measures}, the random set $\xre{K\times L}\es{E^c}$ is $\eps$-rare with respect to $\eta(\otimes_{i=1}^n\mu_i)$.
    By the bijection $m:\Z^2\to G$, $(x,y)\mapsto(\frac{K}{k}x+\frac{x_0}{k},\frac{L}{l}y+\frac{y_0}{l})$, $B\cap G$ is $\eps$-rare with respect to $\otimes_{i=1}^n\mu_i$, so that
    \begin{equation}
        (\otimes_{i=1}^n\mu_i)(A\subset B)
        \le(\otimes_{i=1}^n\mu_i)(A\cap G\subset B\cap G)
        \le\eps^{\abs{A\cap G}}
        \le\eps^{\abs{A}/r},
    \end{equation}
    as required.
\end{proof}

\section{One-dimensional systems}
\label{sec:one-dimensional systems}

In this section, we compute several quantities related to the partition function of one-dimensional monomer-dimer systems under various boundary conditions, using the transfer matrix method.
Recall that $\lambda\in\R$ and $a>0$ have been fixed such that $\chempot+a>0$; the computations in this section do not require the additional assumption $3a>\chempot$.

We define one-dimensional monomer-dimer systems analogously to their two-dimensional counterparts.
Specifically, we consider $\Z$ in its usual graphical representation as an infinite path, and identify its edge set $\V^{\oned}$ with the set of midpoints of its edges, i.e., $1/2+\Z$.
Given $L\ge 1$, define the one-dimensional rectangle $\mathrm{R}_L:=[-1/2,L-1/2]\subset\R$.
Dimer configurations on $\V^\oned$, boundary conditions on a one-dimensional rectangle $R$, and their associated finite-volume Hamiltonian and Gibbs measure are defined analogously to those in two dimensions and denoted using the same notation with the superscript $\oned$ added.
In particular, the vacancy potential~\eqref{eqn:vacancy potential} becomes a two-body interaction in one dimension, while the broken link potential~\eqref{eqn:broken-link potential} remains three-body.

To construct the transfer matrix, we picture a sliding window spanning a width of two edges in $\V^\oned$ moving from left to right.
There are exactly three possible states within the window: the left edge is occupied by a dimer, the right edge is occupied, or neither one is.
Denote these states by $\ket{l}$, $\ket{r}$, and $\ket{0}$, respectively.
In the $(\ket{0},\ket{l},\ket{r})$-basis, the transfer matrix is
\begin{equation}
    \label{eqn: transfer matrix}
    T=\begin{bmatrix}
        e^{-\frac{1}{2}\beta(\chempot+a)} & 0 & e^{-\frac{1}{2}\beta(\chempot+2a)} \\
        e^{-\frac{1}{2}\beta a} & 0 & 1 \\
        0 & 1 & 0
    \end{bmatrix}.
\end{equation}

\subsection{Expansions for the eigenvalues of $T$}
\label{sec:eigenvalue expansions}

The eigenvalues of $T$ are the roots of the cubic polynomial
\begin{equation}
    \label{eqn:characteristic polynomial}
    p(\eigenval):=\det(\eigenval I-T)
    =\eigenval^3-e^{-\frac{1}{2}\beta(\chempot+a)}\eigenval^2-\eigenval+e^{-\frac{1}{2}\beta(\chempot+a)}-e^{-\frac{1}{2}\beta(\chempot+3a)}.
\end{equation}
To obtain usable expressions for the eigenvalues, we treat the last term of $p(\eigenval)$ as a small perturbation and introduce the auxiliary polynomial
\begin{equation}
    p_\eps(\eigenval):=\eigenval^3-e^{-\frac{1}{2}\beta(\chempot+a)}\eigenval^2-\eigenval+e^{-\frac{1}{2}\beta(\chempot+a)}+\eps,
\end{equation}
where $\eps\in[-e^{-\frac{1}{2}\beta(\chempot+3a)},0]$.
When $\eps=0$, $p_\eps(\eigenval)=(\eigenval-1)(\eigenval-e^{-\frac{1}{2}\beta(\chempot+a)})(\eigenval+1)$ has roots $\eigenval_1(0)=1$, $\eigenval_2(0)=e^{-\frac{1}{2}\beta(\chempot+a)}$, and $\eigenval_3(0)=-1$.
By the implicit function theorem, each root $\eigenval_i(0)$, $i=1,2,3$, extends to an analytic function $\eigenval_i(\eps)$ satisfying $p_\eps(\eigenval_i(\eps))\equiv0$ for small $\abs{\eps}$, which admits the expansion
\begin{equation}
    \eigenval_i(\eps)=\eigenval_i(0)+\eigenval_i'(0)\eps+\order{\eps^2}.
\end{equation}
Differentiating the identity $p_\eps(\eigenval_i(\eps))=0$ with respect to $\eps$ at $\eps=0$ and solving for $\eigenval_i'(0)$ yields
\begin{equation}
    \eigenval_1'(0)=-\frac{1}{2(1-e^{-\frac{1}{2}\beta(\chempot+a)})},\quad
    \eigenval_2'(0)=\frac{1}{1-e^{-\beta(\chempot+a)}},\quad
    \eigenval_3'(0)=-\frac{1}{2(1+e^{-\frac{1}{2}\beta(\chempot+a)})}.
\end{equation}
Hence, there exists a constant $\beta_0>0$ such that, for all $\beta\ge\beta_0$, the eigenvalues of $T$ admit the expansions
\begin{align}
    \label{eqn:lambda_1 expansion}
    \eigenval_1&{}=1+\frac{1}{2}e^{-\frac{1}{2}\beta(\chempot+3a)}+\order{e^{-\beta(\chempot+2a)}}, \\
    \label{eqn:lambda_2 expansion}
    \eigenval_2&{}=e^{-\frac{1}{2}\beta(\chempot+a)}-e^{-\frac{1}{2}\beta(\chempot+3a)}+\order{e^{-\beta(\chempot+2a)}}, \\
    \label{eqn:lambda_3 expansion}
    \eigenval_3&{}=-1+\frac{1}{2}e^{-\frac{1}{2}\beta(\chempot+3a)}+\order{e^{-\beta(\chempot+2a)}}.
\end{align}

\subsection{The partition function of one-dimensional systems}
\label{sec:one-dimensional partition function bounds}

Using the eigenvalue expansions derived in Section~\ref{sec:eigenvalue expansions}, we now fulfill the goal of this subsection.
We start by defining a length scale that will play a recurrent role in this paper.
Recall that the \textbf{correlation length}~\cite[Section 15.1]{takahashi1999thermodynamics} of the one-dimensional monomer-dimer system is defined as 
\begin{equation}
    \xi:=\frac{1}{\log\abs{\eigenval_1/\eigenval_3}}=\frac{1}{\log(1+e^{-\frac{1}{2}\beta(\chempot+3a)}+\order{e^{-\beta(\chempot+2a)}})}.
\end{equation}
We use the leading-order term of $\xi$ to define the length scale
\begin{equation}
    \ell_0:=e^{\frac{1}{2}\beta(\chempot+3a)}.
\end{equation}

The first quantity we estimate is, up to a factor of $e^{-\frac{1}{2}\beta(\chempot+a)}$, the partition function on a rectangle of even length with vacant boundary conditions.

\begin{proposition}
    \label{prop: partition function vacant boundary}
    For all $K>0$, there exist constants $\beta_K,C_K>0$ such that, for all $\beta\ge\beta_K$ and odd integers $q$ such that $3\le q\le K\ell_0$, 
    \begin{equation}
        \label{eqn: partition function vacant boundary - expansion}
        \abs{\bra{0}T^q\ket{0}-e^{-\frac{1}{2}\beta(\chempot+3a)}\cosh(\frac{q}{2\ell_0})-e^{-\frac{1}{2}\beta(\chempot+a)q}}
        \le C_Ke^{-\beta(\chempot+2a)}.
    \end{equation}
\end{proposition}

\begin{proof}
    It is straightforward to compute a set of eigenvectors of $T$, which allows us to diagonalize $T=PDP^{-1}$, where $D=\mathrm{diag}(\eigenval_1,\eigenval_2,\eigenval_3)$, and 
    \begin{equation}
        P=\begin{bmatrix}
            e^{-\frac{1}{2}\beta(\chempot+2a)} & e^{-\frac{1}{2}\beta(\chempot+2a)} & e^{-\frac{1}{2}\beta(\chempot+2a)} \\
            \eigenval_1(\eigenval_1-e^{-\frac{1}{2}\beta(\chempot+a)}) & \eigenval_2(\eigenval_2-e^{-\frac{1}{2}\beta(\chempot+a)}) & \eigenval_3(\eigenval_3-e^{-\frac{1}{2}\beta(\chempot+a)}) \\
            \eigenval_1-e^{-\frac{1}{2}\beta(\chempot+a)} & \eigenval_2-e^{-\frac{1}{2}\beta(\chempot+a)} & \eigenval_3-e^{-\frac{1}{2}\beta(\chempot+a)}
        \end{bmatrix}.
    \end{equation}
    Note that $P^{-1}\ket{0}$ is precisely the first column of $P^{-1}$, which we now compute using the adjugate method.
    By an elementary computation, we find that
    \begin{equation}
        \det(P)=-e^{-\frac{1}{2}\beta(\chempot+2a)}(\eigenval_1-\eigenval_2)(\eigenval_2-\eigenval_3)(\eigenval_3-\eigenval_1),
    \end{equation}
    and that the first column of $P^{-1}$ is given by
    \begin{equation}
        -\frac{1}{e^{-\frac{1}{2}\beta(\chempot+2a)}}
        \begin{bmatrix}
            \frac{(\eigenval_2-e^{-\frac{1}{2}\beta(\chempot+a)})(\eigenval_3-e^{-\frac{1}{2}\beta(\chempot+a)})}{(\eigenval_3-\eigenval_1)(\eigenval_1-\eigenval_2)} \\
            \frac{(\eigenval_3-e^{-\frac{1}{2}\beta(\chempot+a)})(\eigenval_1-e^{-\frac{1}{2}\beta(\chempot+a)})}{(\eigenval_1-\eigenval_2)(\eigenval_2-\eigenval_3)} \\
            \frac{(\eigenval_1-e^{-\frac{1}{2}\beta(\chempot+a)})(\eigenval_2-e^{-\frac{1}{2}\beta(\chempot+a)})}{(\eigenval_2-\eigenval_3)(\eigenval_3-\eigenval_1)}
        \end{bmatrix}
        =-\frac{1}{e^{-\frac{1}{2}\beta(\chempot+2a)}}
        \begin{bmatrix}
            -\frac{1}{2}e^{-\frac{1}{2}\beta(\chempot+3a)}+\order{e^{-\beta(\chempot+2a)}} \\
            -1+\order{e^{-\beta(\chempot+2a)}} \\
            \frac{1}{2}e^{-\frac{1}{2}\beta(\chempot+3a)}+\order{e^{-\beta(\chempot+2a)}}
        \end{bmatrix}.
    \end{equation}
    Hence, 
    \begin{equation}
        \label{eqn:partition function vacant boundary - expansion}
        \begin{multlined}
        \bra{0}T^q\ket{0}=\left[\frac{1}{2}e^{-\frac{1}{2}\beta(\chempot+3a)}+\order{e^{-\beta(\chempot+2a)}}\right]\eigenval_1^q
        +\left[1+\order{e^{-\beta(\chempot+2a)}}\right]\eigenval_2^q
        \\
        +\left[-\frac{1}{2}e^{-\frac{1}{2}\beta(\chempot+3a)}+\order{e^{-\beta(\chempot+2a)}}\right]\eigenval_3^q.
        \end{multlined}
    \end{equation}
    On the one hand, by~\eqref{eqn:lambda_1 expansion} and~\eqref{eqn:lambda_3 expansion} and the assumptions that $q$ is odd and $q\le K\ell_0$, we have that
    \begin{equation}
        \begin{split}
            q\log x_1{}&=\frac{1}{2}e^{-\frac{1}{2}\beta(\chempot+3a)}q+\order{K e^{-\frac{1}{2}\beta(\chempot+a)}},
            \\
            q\log(-x_3){}&=-\frac{1}{2}e^{-\frac{1}{2}\beta(\chempot+3a)}q+\order{K e^{-\frac{1}{2}\beta(\chempot+a)}},
        \end{split}
    \end{equation}
    and so
    \begin{equation}
        x_1^q=e^{q/(2\ell_0)}\left[1+\order{K e^{-\frac{1}{2}\beta(\chempot+a)}}\right],
        \quad
        x_3^q=-e^{-q/(2\ell_0)}\left[1+\order{K e^{-\frac{1}{2}\beta(\chempot+a)}}\right].
    \end{equation}
    On the other hand, by~\eqref{eqn:lambda_2 expansion} and the assumption that $q\ge 3$,
    \begin{equation}
        \begin{multlined}
            \abs{x_2^q-e^{-\frac{1}{2}\beta(\chempot+a)q}}
            \le q e^{-\frac{1}{2}\beta(\chempot+a)(q-1)} \abs{x_2-e^{-\frac{1}{2}\beta(\chempot+a)}}
            \\
            \le \sup_{p\ge 3}\set{p e^{-\frac{1}{2}\beta(\chempot+a)(p-2)}} e^{-\frac{1}{2}\beta(\chempot+a)}\order{e^{-\frac{1}{2}\beta(\chempot+3a)}}
            =\order{e^{-\beta(\chempot+2a)}}.
        \end{multlined}
    \end{equation}
    Combining the above estimates yields~\eqref{eqn: partition function vacant boundary - expansion}.
\end{proof}

\begin{corollary}
    \label{cor: partition function vacant boundary}
    For all $K>0$, there exist constants $\beta_K,C_K>0$ such that, for all $\beta\ge\beta_K$ and odd integers $q$ such that $1\le q\le K\ell_0$, 
    \begin{equation}
        \label{eqn: partition function vacant boundary - bound}
        \bra{0}T^q\ket{0}
        \le C_K\left[e^{-\frac{1}{2}\beta(\chempot+3a)}+e^{-\frac{1}{2}\beta(\chempot+a)q}\right].
    \end{equation}
\end{corollary}

\begin{proof}
    When $q=1$, \eqref{eqn: partition function vacant boundary - bound} follows directly from the definition~\eqref{eqn: transfer matrix} of the transfer matrix.
    Otherwise, \eqref{eqn: partition function vacant boundary - bound} follows from Proposition~\ref{prop: partition function vacant boundary}.
\end{proof}

The remaining results concern the two-dimensional monomer-dimer model and follow from the one-dimensional results above.

\begin{proposition}
    \label{prop:partition function periodic boundary}
    For all $0<c<\frac{1}{2}$, there exists $\beta_0>0$ such that, for all $\beta\ge\beta_0$ and rectangles $\Lambda$ with either even width or even height,
    \begin{equation}
        \label{eqn:partition function periodic boundary}
        Z_{\Lambda;\beta}^{\per}\ge \left[1+ce^{-\frac{1}{2}\beta(\chempot+3a)}\right]^{\Area{\Lambda}}.
    \end{equation}
\end{proposition}

\begin{proof}
    Let $0<c<\frac{1}{2}$.
    By the expansion~\eqref{eqn:lambda_1 expansion}, there exists $\beta_0>0$ such that $\eigenval_1\ge 1+ce^{-\frac{1}{2}\beta(\chempot+3a)}$ for all $\beta\ge\beta_0$.
    Assuming without loss of generality that $\Width{\Lambda}$ is even, we have that
    \begin{equation}
        Z_{\Lambda;\beta}^{\per}
        \ge \parentheses[\Big]{\trace(T^{\Width{\Lambda}})}^{\Height{\Lambda}}
        \ge \eigenval_1^{\Area{\Lambda}}
        \ge \brackets[\Big]{1+ce^{-\frac{1}{2}\beta(\chempot+3a)}}^{\Area{\Lambda}},
    \end{equation}
    as required.
\end{proof}

\begin{corollary}
    \label{cor:local defect chessboard seminorms}
    Let $v$ be a vertex and $e$ be an edge of $\Z^2$, and denote by $R_v$ and $R_e$ the $1\times 1$ and $2\times 1$ or $1\times 2$ rectangles containing $v$ and $e$, respectively.
    Under the setting of Proposition~\ref{prop:partition function periodic boundary},
    \begin{enumerate}
        \item \label{itm:vacancy chessboard seminorm} $\norm{v\text{ is a vacancy}}_{R_v}\le e^{-\frac{1}{2}\beta(\chempot+a)}$,
        \item \label{itm:broken link chessboard seminorm} $\norm{e\text{ is a broken link}}_{R_e}\le 6e^{-\frac{1}{2}\beta a}$.
    \end{enumerate}
\end{corollary}

\begin{proof}
    Item~\ref{itm:vacancy chessboard seminorm} follows immediately from the definition of the chessboard seminorm and Proposition~\ref{prop:partition function periodic boundary}, which we use to lower bound the partition function with periodic boundary conditions by $1$.
    To prove Item~\ref{itm:broken link chessboard seminorm}, we assume without loss of generality that $e$ is horizontal.
    There are six possible dimer configurations on $R_e$ that make $e$ into a broken link, each having $R_e$-chessboard seminorm no greater than $e^{-\frac{1}{2}\beta a}$.
    Item~\ref{itm:broken link chessboard seminorm} then follows by the subadditivity of the chessboard seminorm.
\end{proof}

\section{Characterizations of orientational symmetry breaking}
\label{sec: characterizations of orientational symmetry breaking}

\subsection{The Heilmann--Lieb approach}
\label{sec: Heilmann-Lieb approach}

In this subsection, we recall the Heilmann--Lieb characterization of orientational order in the monomer-dimer system~\cite[Section~5]{heilmann1979lattice} and use it to characterize all $\Z^2$-ergodic Gibbs measures under the strong-percolation framework of Section~\ref{sec:strongly percolating sets}.

Consider the $2\times2$ square $R=\rect{2\times 2}{(-1/2,-1/2)}$. 
We say that $R$ is \textbf{bad} with respect to a dimer configuration $\sigma$ if either (a) it contains a vacancy of $\sigma$, or (b) it does not contain any vacancy but intersects both horizontal and vertical dimers in $\sigma$, and denote by $Q$ the event that $R$ is bad.
It is easy to check by a case analysis that, in the latter case, at least one broken link of $\sigma$ is local to $R$, namely, $R$ contains an edge $e\in\V$ such that the event that $e$ is a broken link is $R$-local.
Combining Lemma~\ref{lem:recursive chessboard estimate} and Corollary~\ref{cor:local defect chessboard seminorms} and applying a union bound, we find that there exists a universal constant $C>0$ such that
\begin{equation}
    \label{eqn: bad square chessboard seminorm}
    \norm{Q}_R \le Ce^{-\beta\min\set{\chempot+a,a}/2}.
\end{equation}

If $R$ is not bad, then every vertex of $\Z^2$ in $R$ is incident to a dimer and these dimers are either all vertical or all horizontal.
Denote these $R$-local events by $E_\ver$ and $E_\hor$, respectively, and introduce the random sets
\begin{equation}
    \Psi_\ver:=\es{E_\ver}=\xre{1\times1}{\es{E_\ver}},\quad
    \Psi_\hor:=\es{E_\hor}=\xre{1\times1}{\es{E_\hor}},\quad
    \text{and}\quad
    \Psi:=\Psi_\ver\cup\Psi_\hor.
\end{equation}
By Item~\ref{itm:strongly percolating grid events - chessboard bound} of Lemma~\ref{lem:strongly percolating grid events}, for all periodic Gibbs measures $\mu$,
\begin{equation}
    \peierls{\mu}(\Psi)\le \sqrt[4]{\norm{Q}_R};
\end{equation}
in particular, by compactness, there exists a $\Z^2$-invariant Gibbs measure $\mu$ which is the weak limit of a sequence of finite-volume Gibbs measures with periodic boundary conditions.
On the other hand, notice that if $v\in\Psi_\ver$ and $w\in\Psi_\hor$, then $v$ and $w$ are not $\square$-adjacent.
Thus, Lemma~\ref{lem:splitting strongly percolating sets} yields, for all $\Z^2$-ergodic Gibbs measures $\mu$,
\begin{equation}
    \min\set{\peierls{\mu}(\Psi_\ver),\peierls{\mu}(\Psi_\hor)}\le\peierls{\mu}(\Psi).
\end{equation}
From here on, we increase $\beta_0$ so that $\eps:=\sqrt[4]{Ce^{-\beta_0\min\set{\chempot+a,a}/2}}<\eps_0/2$, and continue to assume that $\beta\ge\beta_0$.
Then, for all $\Z^2$-ergodic Gibbs measures $\mu$, 
\begin{equation}
    \label{eqn: at least one good event strongly percolates}
    \min\set{\peierls{\mu}(\Psi_\ver),\peierls{\mu}(\Psi_\hor)}\le\eps<\eps_0/2.
\end{equation}

\begin{proposition}[Two phases]
    \label{prop: two phases}
    Suppose that $\beta\ge\beta_0$.
    Let $\mu$ be a $\Z^2$-ergodic Gibbs measure.
    There exists a unique label $\ast\in\set{\ver,\hor}$ such that, out of the random sets $\Psi_\ver$ and $\Psi_\hor$, $\Psi_\ast$ is $\eps$-strongly percolating, while the other one has only finite $\boxtimes$-components $\mu$-a.s.
\end{proposition}

\begin{proof}
    Let $\mu$ be a $\Z^2$-ergodic Gibbs measure.
    In view of~\eqref{eqn: at least one good event strongly percolates}, we may assume without loss of generality that $\peierls{\mu}(\Psi_\ver)\le\eps$, as the other case is analogous.
    We need to show that $\Psi_\hor$ has finite $\boxtimes$-components $\mu$-a.s.
    Fix $\eps'\in(\eps,\eps_0/2)$.
    By~\eqref{eqn: strong percolation parameter} and Lemma~\ref{lem:Peierls argument}, there exists an $\eps'$-rare set $B$ such that $\Psi_\ver$ contains the unique infinite $\square$-component $I$ of $\Z^2\setminus B$, and each $\boxtimes$-component of $\Z^2\setminus I$ is finite, $\mu$-a.s.
    Noting that $\Psi_\hor\subset\Z^2\setminus\Psi_\ver\subset\Z^2\setminus I$ completes the proof.
\end{proof}

For a $\Z^2$-ergodic Gibbs measure $\mu$, let $\Phase{\mu}$ denote the label guaranteed by Proposition~\ref{prop: two phases}.
Let $\tau:\R^2\to\R^2$ be the reflection $(x,y)\mapsto(y,x)$, and denote by $\tau\mu$ the push-forward of $\mu$ under $\tau$.
It is immediate that $\Phase{\mu}=\ver$ if and only if $\Phase{\tau\mu}=\hor$~\cite[Proposition~7.2(3)]{hadas2025columnar}.

\subsection{Microscopic rarity of horizontal dimers in the vertical phase}
\label{sec: microscopic rarity}

In this subsection, we strengthen the Heilmann--Lieb characterization of the $\Z^2$-ergodic Gibbs measures by establishing the strong percolation of a larger random set that, intuitively, captures the microscopic rarity of horizontal dimers in the vertical phase, and vice versa.

Recall that we have been working under the assumptions that $a>0$, $\chempot+a>0$, and $3a>\chempot$.
Introduce the \textit{$\beta$-independent} microscopic length scales
\begin{equation}
    \label{eqn: microscopic length scales}
    \fm:=\max\set{24,\ceil{\frac{96a}{\chempot+a}}},\quad
    \fn:=\max\set{4(\fm+1),2\floor{\frac{12a}{3a-\chempot}}+2}.
\end{equation}
Consider the $\fn\times\fn$ square $R=\rect{\fn\times \fn}{(-1/2,-1/2)}$.
Let $\fe$ denote the horizontal edge of $\Z^2$ with endpoints $(\fn/2-1,\fn/2)$ and $(\fn/2,\fn/2)$, and define the $R$-local event
\begin{equation}
    \label{eqn: weak vertical order}
    F_\ver:=\set{\sigma\in\Omega\mid\sigma(\fe)=0}
\end{equation}
and the random set $\Theta_\ver:=\xre{\fn\times \fn}\es{F_\ver}$.
The main result of this subsection is as follows.

\begin{theorem}
    \label{thm: microscopic rarity of horizontal dimers}
    There exists a constant $C_\fn>0$ such that the following holds.
    If $\beta\ge\beta_0$ and $\mu$ is a $\Z^2$-ergodic Gibbs measure such that $\Phase{\mu}=\ver$, then
    \begin{equation}
        \label{eqn: microscopic rarity of horizontal dimers}
        \peierls{\mu}(\Theta_\ver)\le C_\fn e^{-3\beta(1-1/\fn)a}.
    \end{equation}
\end{theorem}

The rest of this subsection is dedicated to the proof of Theorem~\ref{thm: microscopic rarity of horizontal dimers}.
We will proceed by explicitly constructing an $\order{e^{-3\beta(1-1/\fn)a}}$-rare set $B$ such that $\Theta_\ver$ contains the unique infinite $\square$-component of $\Z^2\setminus B$, $\mu$-a.s.

\subsubsection{Construction of the rare set}
\label{sec: construction of the rare set}

In this sub-subsection, we construct the random set $B$.
Let $\mu$ be a $\Z^2$-ergodic Gibbs measure such that $\Phase{\mu}=\ver$.
The following construction is valid $\mu$-a.s.

Sample $\sigma\sim\mu$.
Let $V(\sigma)\subset\Z^2$ consist of those vertices which are either (a) a vacancy of $\sigma$, or (b) incident to a horizontal dimer in $\sigma$.
By Proposition~\ref{prop: two phases}, $\peierls{\mu}(\Psi_\ver)\le\eps$, so $V(\sigma)$ has only finite $\boxtimes$- and hence $\square$-components.
Denote by $\cC(\sigma)$ the set of those $\square$-components $C$ of $V(\sigma)$ with the following property: there exists $(x,y)\in\Z^2\setminus\Theta_\ver(\sigma)$ such that both endpoints of the translated edge $\eta_{(\fn x,\fn y)}\fe$ belong to $C$.
For each $C\in\cC(\sigma)$, define the $\fn$-coarse-grained version of $C$ as 
\begin{equation}
    D_C:=\set{(x,y)\in\Z^2\mid \rect{\fn\times \fn}{(\fn x-1/2,\fn y-1/2)}\cap C\ne\emptyset}.
\end{equation}
Since $C$ is finite, so is $D_C$, and we may let $I_C$ be the unique infinite $\square$-component of $\Z^2\setminus D_C$ and $H_C:=\Z^2\setminus I_C$ be the hull of $D_C$.
Define
\begin{equation}
    B_C:=\set{(x,y)\in H_C\mid\text{there exists }(x',y')\in I_C\text{ that is $\square$-adjacent to }(x,y)},
\end{equation}
and let $B=B(\sigma):=\bigcup_{C\in\cC(\sigma)}B_C$.
We record the following simple topological property of $B$.

\begin{lemma}
    \label{eqn: rare set topological property}
    If $\Z^2\setminus B(\sigma)$ has a unique infinite $\square$-component $I$, then $I\subset\Theta_\ver(\sigma)$.
\end{lemma}

\begin{proof}
    Under the assumption of the lemma, we must have that $I\subset Z^2\setminus\bigcup_{C\in\cC(\sigma)}H_C$, so
    \begin{equation}
        \bigcup_{C\in\cC(\sigma)}H_C\subset\Z^2\setminus I.
    \end{equation}
    On the other hand, it follows from the definition of $\cC(\sigma)$ that 
    \begin{equation}
    \Z^2\setminus\Theta_\ver(\sigma)\subset\bigcup_{C\in\cC(\sigma)}D_C\subset\bigcup_{C\in\cC(\sigma)}H_C.
    \end{equation}
    Combining these inclusions completes the proof.
\end{proof}

\subsubsection{Defects associated with the rare set}
\label{sec: defects associated with the rare set}

We continue to work in the setting of Section~\ref{sec: construction of the rare set}.
To quantify the rarity of the random set $B=B(\sigma)$ constructed there, we will associate a set of defects with each $\square$-component $C\in\cC(\sigma)$ in such a way that distinct $\square$-components have disjoint sets of associated defects.

Given $e\in\V$ and a dimer configuration $\sigma$, we say that $e$ is a \textbf{defect edge} of $\sigma$ if either (a) $e$ is a broken link of $\sigma$, or (b) $e$ is not a broken link of $\sigma$ but is incident to a vacancy of $\sigma$.

Let $\sigma$ be a dimer configuration for which $B(\sigma)$ is defined and $C\in\cC(\sigma)$.
Since $C$ is finite, we may let $I^\mic_C$ be the unique infinite $\square$-component of $\Z^2\setminus C$ and $H^\mic_C:=\Z^2\setminus I^\mic_C$ be the hull of $C$.
Define
\begin{equation}
    \gamma_C:=\set{\set{v,w}\in\V\mid v\in H^\mic_C,\, w\in I^\mic_C,\, v\text{ is $\square$-adjacent to }w}.
\end{equation}
By the definition of $V(\sigma)$ of which $C$ is a $\square$-component, every $e\in\gamma_C$ is a defect edge, and $\gamma_{C}\cap\gamma_{C'}=\emptyset$ for distinct $\square$-components $C,C'\in\cC(\sigma)$.

\begin{lemma}
    \label{lem: a boundary edge in every boundary box}
    For each $(x,y)\in B_C$, there exists an edge in $\gamma_C$ that intersects $\rect{\fn\times \fn}{(\fn x-1/2,\fn y-1/2)}$.
\end{lemma}

\begin{proof}
    Let $(x',y')\in I_C$ be $\square$-adjacent to $(x,y)$.
    Then, $\rect{\fn\times\fn}{(\fn x'-1/2,\fn y'-1/2)}\cap\Z^2\subset I_C^\mic$.
    Fix $v\in C\cap\rect{\fn\times\fn}{(\fn x-1/2,\fn y-1/2)}$ and $w\in\rect{\fn\times\fn}{(\fn x'-1/2,\fn y'-1/2)}\cap\Z^2$ arbitrarily, and consider a $\square$-path from $v$ to $w$ in $\rect{\fn\times\fn}{(\fn x-1/2,\fn y-1/2)}\cup \rect{\fn\times\fn}{(\fn x'-1/2,\fn y'-1/2)}$.
    This path contains the desired edge.
\end{proof}

For each $(x,y)\in B_C$, fix an edge $e_{C}(x,y)$ as supplied by Lemma~\ref{lem: a boundary edge in every boundary box}.

Let $\gamma_C^\ast$ be the subgraph of $(\Z^2)^\ast$ consisting of the dual edges crossing the edges of $\gamma_C$.
Then, $\gamma_C^\ast$ is a connected subgraph of $(\Z^2)^\ast$ of even degree, so it admits an Eulerian cycle.
With respect to this Eulerian cycle, we partition $\gamma_C^\ast$ into sets of consecutive dual edges, which we call \textbf{arcs}, as follows:
\begin{enumerate}
    \item If $\abs{\gamma_C}<\fm$, then we take $\gamma_C^\ast$ itself as the only arc.
    \item If $\abs{\gamma_C}\ge\fm$, then we take $\floor{\abs{\gamma_C}/\fm}-1$ arcs of length $\fm$ and one arc of length $\abs{\gamma_C}-\fm(\floor{\abs{\gamma_C}/\fm}-1)$.
    Thus, every arc in this case has length between $\fm$ and $2\fm-1$.
\end{enumerate}
Define a map $\pi_C$ from $B_C$ to the set of arcs of $\gamma_C^\ast$ by mapping each $(x,y)$ to the unique arc containing the dual edge crossing $e_{C}(x,y)$.

\begin{lemma}
    \label{lem: cardinality of arc preimage}
    The preimage of each arc of $\gamma_C^\ast$ under $\pi_C$ has cardinality at most four.
\end{lemma}

\begin{proof}
    Since $\fn\ge 4\fm+4$ and the length of an arc is at most $2\fm-1$, each arc intersects the $\fn\times\fn$ square $\rect{\fn\times\fn}{(\fn x-1/2,\fn y-1/2)}$ for at most four distinct $(x,y)\in\Z^2$.
\end{proof}

The main result of this sub-subsection is as follows.

\begin{proposition}[Localization of defect edges]
    \label{prop: localization of defect edges}
    Let $\sigma\in\Omega$ be a dimer configuration for which $B(\sigma)$ is defined and $A\subset\Z^2$ be a finite set such that $A\subset B(\sigma)$.
    For each $C\in\cC(\sigma)$ such that $A_C:=A\cap B_C\ne\emptyset$, the following hold:
    \begin{enumerate}
        \item \label{itm: short arcs} If $\abs{\gamma_C}<\fm$, then $B_C=\set{(x,y)}$ for some $(x,y)\in\Z^2$, and $\gamma_C$ contains six broken links of $\sigma$, all local to $\rect{\fn\times\fn}{(\fn x-1/2,\fn y-1/2)}$.
        \item \label{itm: long arcs} If $\abs{\gamma_C}\ge\fm$, then $\gamma_C$ contains at least $\fm\abs{A_C}/4$ defect edges of $\sigma$, all local to
        \begin{equation}
            \bigcup_{(x,y)\in A_C}\rect{3\fn\times 3\fn}{(\fn(x-1)-1/2,\fn(y-1)-1/2)}.
        \end{equation}
    \end{enumerate}
    In particular, the defect edges identified above are distinct for distinct $\square$-components $C$ with $A_C\ne\emptyset$.
\end{proposition}

\begin{proof}
    Let $C\in\cC(\sigma)$ be such that $A_C\ne\emptyset$.

    First, suppose that $\abs{\gamma_C}<\fm$.
    By the definition of $\cC(\sigma)$, $C$ contains the vertices of the translated edge $\eta_{(\fn x,\fn y)}\fe$ for some $(x,y)\in\Z^2$.
    Since $\fn\ge 4\fm+4$, this forces $B_C=D_C=\set{(x,y)}$.
    Now, starting from $\eta_{(\fn x,\fn y)}\fe$, the horizontal edges of $\Z^2$ to its left and right alternate between being covered and not being covered by a horizontal dimer in $\sigma$, and the pattern necessarily terminates with a broken link local to $\rect{\fn\times\fn}{(\fn x-1/2,\fn y-1/2)}$.
    Similarly, starting from either endpoint of $\eta_{(\fn x,\fn y)}\fe$, the first vertices of $\Z^2$ above and below it that are covered by vertical dimers in $\sigma$ are each incident to a vertical broken link local to $\rect{\fn\times\fn}{(\fn x-1/2,\fn y-1/2)}$.
    Altogether, we have identified six broken links of $\sigma$ with the desired localization property.

    Suppose now that $\abs{\gamma_C}\ge\fm$.
    By Lemma~\ref{lem: cardinality of arc preimage}, the image of $A_C$ under $\pi_C$ has cardinality at least $\abs{A_C}/4$.
    Moreover, every arc in $\pi_C(A_C)$ has length at least $\fm$, with every dual edge in the arc crossing a defect edge of $\sigma$, so $\gamma_C$ contains at least $\fm\abs{A_C}/4$ defect edges of $\sigma$.
    For the remaining localization property, recall that for each $(x,y)\in A_C$, the arc $\pi_C(x,y)$ contains the dual edge crossing $e_C(x,y)$, which intersects $\rect{\fn\times\fn}{(\fn x-1/2,\fn y-1/2)}$.
    Since $\fn\ge4\fm+4$ and the arc in question has length at most $2\fm-1$, we conclude that all the defect edges of $\sigma$ crossed by the arc are local to $\rect{3\fn\times 3\fn}{(\fn(x-1)-1/2,\fn(y-1)-1/2)}$.

    The final claim follows from the fact that $\gamma_C\cap\gamma_{C'}=\emptyset$ for distinct $C,C'\in\cC(\sigma)$.
\end{proof}

\subsubsection{Proof of Theorem~\ref{thm: microscopic rarity of horizontal dimers}}

By Lemma~\ref{eqn: rare set topological property}, it suffices to show that $B$ is $C_\fn e^{-3\beta(1-1/\fn)a}$-rare for some $\beta$-independent constant $C_\fn>0$, uniformly in the choice of the $\Z^2$-ergodic Gibbs measure $\mu$ such that $\Phase{\mu}=\ver$.

Let $A\subset\Z^2$ be a finite set. 
Rather than dealing with the event $\set{A\subset B}$ directly, we will first use Proposition~\ref{prop: localization of defect edges} to cover it by a finite union of intersections of local events enforcing prescribed numbers of defect edges, and then apply Proposition~\ref{prop:infinite-volume chessboard estimate for product measures} to the intersections.
A technical complication, absent for the vacant-face event used by Hadas--Peled~\cite[Corollary~4.5]{hadas2025columnar}, is that the support of a defect-edge event has a nontrivial shape (see Item~\ref{itm:broken link chessboard seminorm} of Corollary~\ref{cor:local defect chessboard seminorms}).
This issue manifests itself in the following fashion: for every rectangle $R$, there exists a dimer configuration $\sigma$ containing a defect edge $e\in\V$ that straddles the grid $G^R$, so the $R$-chessboard seminorm is undefined for any event that requires $e$ to be a defect edge.
Thus, it is not possible to use Proposition~\ref{prop:infinite-volume chessboard estimate for product measures} to capture \textit{all} of the penalty associated with the defects in a dimer configuration.
The way out that we adopt is to consider, for a given dimer configuration $\sigma$, and with $R=\rect{\fn\times\fn}{(-1/2,-1/2)}$ chosen at the beginning of Section~\ref{sec: microscopic rarity}, the $\fn^2$ distinct translations of the grid $G^R$, one of which must retain at least a fraction of $1-1/\fn$ of the penalty associated with the defects in $\sigma$ as local to its faces.
Thus, Proposition~\ref{prop: localization of defect edges} becomes applicable to these local events, and we can conclude the proof by a union bound.
This turns out to be sufficient for our purposes, provided that $\fn$ is taken sufficiently large in a way that depends on $a$ and $\chempot$, but, importantly, not on $\beta$.

Recall the definition of a defect edge from Section~\ref{sec: defects associated with the rare set}.
When the nature of a defect edge is unspecified, we associate with it an energetic penalty of $(\chempot+a)/8$.
Since every vacancy in a dimer configuration $\sigma$ is incident to exactly four edges of $\Z^2$ and $a/2>(\chempot+a)/8$ by the assumption that $3a>\chempot$, summing the energetic penalties associated with distinct broken links and defect edges (of unspecified nature) in $\sigma$ underestimates the true energetic penalty associated with $\sigma$ via its vacancies and broken links, as computed by the Hamiltonian~\eqref{eqn:finite-volume Hamiltonian}.

Let $\sigma\in\Omega$ be a dimer configuration such that $B(\sigma)$ is defined and $A\subset B(\sigma)$.
Define
\begin{equation}
    \cC_{<,A}(\sigma):=\set{C\in\cC(\sigma)\mid A\cap B_C\ne\emptyset,\,\abs{\gamma_C}<\fm},
\end{equation}
and let $\cE_{<,A}(\sigma)$ be the union of the sets of broken links of $\sigma$ associated with $C\in\cC_{<,A}(\sigma)$ via Item~\ref{itm: short arcs} of Proposition~\ref{prop: localization of defect edges}.
Thus,
\begin{equation}
    \label{eqn: links associated with short arcs cardinality}
    \abs{\cE_{<,A}(\sigma)}=6\abs{\cC_{<,A}(\sigma)}.
\end{equation}
Define $\cC_{\ge,A}(\sigma)$ and $\cE_{\ge,A}(\sigma)$ analogously so that the latter is a set of defect edges of $\sigma$ identified via Item~\ref{itm: long arcs} of Proposition~\ref{prop: localization of defect edges}, with
\begin{equation}
    \label{eqn: links associated with long arcs cardinality}
    \abs{\cE_{\ge,A}(\sigma)}\ge \frac{\fm}{4}\sum_{C\in\cC_{\ge,A}(\sigma)}\abs{A\cap B_C}.
\end{equation}
Note that the defect edges in $\cE_{<,A}(\sigma)\cup\cE_{\ge,A}(\sigma)$ are local to $\bigcup_{(x,y)\in A}\rect{3\fn\times 3\fn}{(\fn(x-1)-1/2,\fn(y-1)-1/2)}$.

Next, with $R=\rect{\fn\times \fn}{(-1/2,-1/2)}$, consider the set of translations of the grid of $R$,
\begin{equation}
    \cG^R:=\set{\eta_{(i,j)}G^R\mid 0\le i,j\le \fn-1}.
\end{equation}
For each $G\in\cG^R$ and $e\in\V$, let $\indicator{G}(e)$ denote the indicator function that both endpoints of $e$ belong to the same face of $G$; thus, if $\indicator{G}(e)=1$, then the event that $e$ is a defect edge is local to some face of $G$.
Note the identity
\begin{equation}
    \label{eqn: shift identity}
    \frac{1}{\fn(\fn-1)}\sum_{G\in\cG^R}\indicator{G}(e)=1,\quad\text{for all }e\in\V.
\end{equation}
Hence, there exists $G\in\cG^R$ such that
\begin{align}
    {}&\frac{a}{2}\sum_{e\in\cE_{<,A}(\sigma)}\indicator{G}(e)
    +\frac{\chempot+a}{8}\sum_{e\in\cE_{\ge,A}(\sigma)}\indicator{G}(e)
    {}&&
    \nonumber
    \\
    \ge{}& \frac{\fn(\fn-1)}{\fn^2}\brackets[\bigg]{\frac{a}{2}\abs{\cE_{<,A}(\sigma)}+\frac{\chempot+a}{8}\abs{\cE_{\ge,A}(\sigma)}}
    {}&&\text{having used~\eqref{eqn: shift identity} to choose $G$}
    \nonumber
    \\
    \ge{}& \bigg(1-\frac{1}{\fn}\bigg)\brackets[\bigg]{3a\abs{\cC_{<,A}(\sigma)}+\frac{\fm(\chempot+a)}{32}\sum_{C\in\cC_{\ge,A}(\sigma)}\abs{A\cap B_C}}
    {}&&\text{using~\eqref{eqn: links associated with short arcs cardinality} and~\eqref{eqn: links associated with long arcs cardinality}}
    \nonumber
    \\
    \ge{}& 3a\bigg(1-\frac{1}{\fn}\bigg)\brackets[\bigg]{\abs{\cC_{<,A}(\sigma)}+\sum_{C\in\cC_{\ge,A}(\sigma)}\abs{A\cap B_C}}
    {}&&\text{by the choice of $\fm$ in~\eqref{eqn: microscopic length scales}}
    \nonumber
    \\
    \ge{}& 3a\bigg(1-\frac{1}{\fn}\bigg)\abs{A}.
    {}&&
    \label{eqn: penalty retainment}
\end{align}

We are now ready to construct the promised cover of the event $\set{A\subset B}$.
For each $G\in\cG^R$, let $\cF_A(G)$ be the set of faces of $G$ that intersect $\bigcup_{(x,y)\in A}\rect{3\fn\times 3\fn}{(\fn(x-1)-1/2,\fn(y-1)-1/2)}$. 
Since any given $\fn\times\fn$ square intersects at most four faces of $G$, we have that
\begin{equation}
    \label{eqn: number of faces intersecting A}
    \abs{\cF_A(G)}\le 36\abs{A}.
\end{equation}
For $\sigma\in\Omega$ and $F\in\cF_A(G)$, define
\begin{align}
    b_F(\sigma){}&:=\abs{\set{e\in\V\mid\text{both endpoints of $e$ belong to $F$, and $e$ is a broken link of $\sigma$}}},
    \\
    d_F(\sigma){}&:=\abs{\set{e\in\V\mid\text{both endpoints of $e$ belong to $F$, and $e$ is a defect edge of $\sigma$}}},
\end{align}
and note that $0\le b_F(\sigma)\le d_F(\sigma)\le 2\fn(\fn-1)$.
For integers $r,s\ge 0$ with $r+s\le 2\fn(\fn-1)$, define the $F$-local event
\begin{equation}
    E_F(r,s):=\set{\sigma\in\Omega\mid b_F(\sigma)\ge r,\,d_F(\sigma)\ge r+s}.
\end{equation}
Also, define $\cP_A(G)$ as the set of functions $\cF_A(G)\to\set{(r,s)\mid r,s\ge 0,\, r+s\le 2\fn(\fn-1)}$ such that, writing $(r_F,s_F)$ for the image of $F\in\cF_A(G)$ under such a function, we have
\begin{equation}
    \sum_{F\in\cF_A(G)}\bigg(\frac{a}{2}\cdot r_F+\frac{\chempot+a}{8}\cdot s_F\bigg)\ge 3a\bigg(1-\frac{1}{\fn}\bigg)\abs{A}.
\end{equation}
By~\eqref{eqn: penalty retainment}, we have the cover
\begin{equation}
    \set{A\subset B}
    \subseteq\bigcup_{G\in\cG^R}\bigcup_{\phi\in\cP_A(G)}\bigcap_{F\in\cF_A(G)}E_F(r_F,s_F).
\end{equation}

It is straightforward to verify that, under the assumption that $3a>\chempot$, there exists a constant $C>0$, depending only on $\fn$ (which is fixed since we have fixed $\chempot$ and $a$), such that, uniformly in $\eta\in\set{\eta_{(i,j)}\mid 0\le i,j\le\fn-1}$, the face $F$ of the translated grid $\eta G^R$, and $r,s\ge 0$,
\begin{equation}
    \norm{E_F(r,s)}_{\eta R}
    \le C\exp{-\beta\bigg(\frac{a}{2}\cdot r+\frac{\chempot+a}{8}\cdot s\bigg)}.
\end{equation}
Recalling~\eqref{eqn: number of faces intersecting A} and using Proposition~\ref{prop:infinite-volume chessboard estimate for product measures}, we conclude the proof by a union bound.

\section{Absence of translational symmetry breaking}
\label{sec: absence of translational symmetry breaking}

\subsection{Disagreement percolation}
\label{sec:disagreement percolation}

In this subsection, we prove the uniqueness of Gibbs measures exhibiting each phase identified in Section~\ref{sec: Heilmann-Lieb approach}. 
As a consequence, we rule out translational symmetry breaking and obtain decay of correlations estimates within each phase, completing the proof of Items~\ref{itm:main - invariance and extremality} and~\ref{itm:main - decay of correlations} of Theorem~\ref{thm:main}.
The main result of this subsection is as follows.

\begin{proposition}
    \label{prop:unique ergodic Gibbs measure for each phase}
    For each $\ast\in\set{\ver,\hor}$, there exists a unique $\Z^2$-ergodic Gibbs measure $\mu_\ast$ such that $\Phase{\mu_\ast}=\ast$. 
    Moreover,
    \begin{enumerate}
        \item \label{itm:unique ergodic Gibbs measure for each phase - extremal} $\mu_\ast$ is extremal;
        \item \label{itm:unique ergodic Gibbs measure for each phase - reflection} $\mu_\hor$ is the push-forward of $\mu_\ver$ under the reflection $(x,y)\mapsto(y,x)$;
        \item \label{itm:unique ergodic Gibbs measure for each phase - decomposition} every periodic Gibbs measure is a convex combination of $\mu_\ver$ and $\mu_\hor$.
    \end{enumerate}
\end{proposition}

The proof of Proposition~\ref{prop:unique ergodic Gibbs measure for each phase} is by implementing the disagreement percolation method introduced by van den Berg in~\cite{van1993uniqueness}.
We quote the following formulation of the method from~\cite{hadas2025columnar}.

\begin{lemma}[{Disagreement percolation,~\cite[Theorem 8.2]{hadas2025columnar}}]
    \label{lem:disagreement percolation}
    Let $\mu,\mu'$ be Gibbs measures of a Markov random field. 
    Let $\sigma,\sigma'$ be independent samples from $\mu$ and $\mu'$, respectively, and define their \textbf{disagreement set} by $\Delta_{\sigma,\sigma'}:=\set{e\in\V\mid\sigma(e)\ne\sigma'(e)}$.
    If 
    \begin{equation}
        \label{eqn:disagreement percolation - no infinite component}
        (\mu\otimes\mu')(\Delta_{\sigma,\sigma'}\text{ has an infinite connected component})=0,
    \end{equation}
    then the following hold:
    \begin{enumerate}
        \item $\mu=\mu'$ and $\mu$ is extremal;
        \item Let $A,B\subset\V$ and $f,g:\Omega\to[-1,1]$ be $A$- and $B$-local functions, respectively.
        Then,
        \begin{equation}
            \abs{\Cov_\mu(f,g)}\le 2(\mu\otimes\mu')(A\text{ and }B\text{ are connected by a path in }\Delta_{\sigma,\sigma'}).
        \end{equation}
    \end{enumerate}
\end{lemma}

Recall from Section~\ref{sec:basic definitions} that the monomer-dimer model is a Markov random field under the $\ddag$-connectivity on $\V$ introduced there.
To verify the criterion~\eqref{eqn:disagreement percolation - no infinite component} of Lemma~\ref{lem:disagreement percolation}, we will prove the following proposition.

\begin{proposition}
    \label{prop:disagreement percolation decay}
    There exist constants $\beta_0,C,c>0$ such that, for all $\beta\ge\beta_0$, if $\mu,\mu'$ are $\Z^2$-ergodic Gibbs measures with $\Phase{\mu}=\Phase{\mu'}=\ver$ and $\sigma,\sigma'$ are independent samples from $\mu$ and $\mu'$, respectively, then, for each $A,B\subset\V$,
    \begin{equation}
        (\mu\otimes\mu')(A\text{ and }B\text{ are connected by a $\ddag$-path in }\Delta_{\sigma,\sigma'})\le \sum_{u\in A}\sup_{v\in B}\alpha_1(u,v),
    \end{equation}
    where 
    \begin{equation}
        \label{eqn: Heilmann-Lieb decay}
        \alpha_1(u,v):=C\exp{-c\abs{x_u-x_v}-c\ell_0^{-1}\abs{y_u-y_v}}.
    \end{equation}
\end{proposition}

Once Proposition~\ref{prop:disagreement percolation decay} is established, Proposition~\ref{prop:unique ergodic Gibbs measure for each phase} readily follows.

\begin{proof}[Proof of Proposition~\ref{prop:unique ergodic Gibbs measure for each phase}]
    The existence of a $\Z^2$-ergodic Gibbs measure $\mu$ with $\Phase{\mu}=\ver$ follows from the discussion in Section~\ref{sec: Heilmann-Lieb approach}.
    To prove uniqueness, suppose that $\mu$ and $\mu'$ are two such measures, and let $\sigma,\sigma'$ be independent samples from $\mu$ and $\mu'$, respectively.
    By Proposition~\ref{prop:disagreement percolation decay}, the condition~\eqref{eqn:disagreement percolation - no infinite component} of Lemma~\ref{lem:disagreement percolation} is satisfied, so $\mu=\mu'$ and $\mu$ is extremal.

    Now, suppose that $\mu$ is a $\Z^2$-ergodic Gibbs measure with $\Phase{\mu}=\hor$.
    As noted at the end of Section~\ref{sec: Heilmann-Lieb approach}, with $\tau$ being the reflection $(x,y)\mapsto(y,x)$, $\Phase{\tau\mu}=\ver$.
    By the uniqueness and extremality of $\mu_\ver$ just proven, $\mu=\tau\mu_\ver$ is also the unique measure with phase label $\hor$ and is extremal.

    It remains to prove Item~\ref{itm:unique ergodic Gibbs measure for each phase - decomposition}.
    Let $\mu$ be a periodic Gibbs measure and $\cL\subset\Z^2$ be a full-rank sublattice such that $\mu$ is $\cL$-invariant.
    By the ergodic decomposition theorem~\cite[Theorem (14.17)]{georgii2011gibbs}, $\mu$ is a convex combination of $\cL$-ergodic Gibbs measures.
    Thus, it suffices to show that each $\cL$-ergodic Gibbs measure is either $\mu_\ver$ or $\mu_\hor$.
    To this end, let $\nu$ be an $\cL$-ergodic Gibbs measure.
    Consider the $\Z^2$-invariant measure 
    \begin{equation}
        \tilde{\nu}:=\frac{1}{\abs{\Z^2/\cL}}\sum_{\eta\in\Z^2/\cL}\eta\nu.
    \end{equation}
    By the ergodic decomposition theorem and the uniqueness proven above, $\tilde{\nu}$ is a convex combination of $\mu_\ver$ and $\mu_\hor$.
    However, noting that $\mu_\ver$ and $\mu_\hor$ are $\cL$-ergodic by extremality and $\eta\nu$ is $\cL$-ergodic for each $\eta\in\Z^2/\cL$ since the latter translation preserves $\cL$-ergodicity, we now have two $\cL$-ergodic decompositions of $\tilde{\nu}$.
    By the uniqueness of the decomposition, we conclude that each $\eta\nu$ is one of $\mu_\ver$ or $\mu_\hor$.
    In particular, $\nu$ is one of $\mu_\ver$ or $\mu_\hor$, which completes the proof.
\end{proof}

The rest of the subsection is dedicated to the proof of Proposition~\ref{prop:disagreement percolation decay}.

\subsubsection{Sealed rectangles}
\label{sec:sealed rectangles}

We now leverage the strong percolation characterization of \emph{a} vertical phase to construct \textbf{sealed rectangles}, which are rectangles in which a joint configuration $(\sigma,\sigma')$ sampled from $\mu\otimes\mu'$, where $\mu,\mu'$ are $\Z^2$-ergodic Gibbs measures with $\Phase{\mu}=\Phase{\mu'}=\ver$, exhibits properties that rule out the presence of large disagreement $\ddag$-components nearby.
Due to the nature of the interactions in the monomer-dimer model, we require a different definition of sealed rectangles from that in~\cite{hadas2025columnar}.

Introduce a \textit{$\beta$-dependent} mesoscopic length scale
\begin{equation}
    \label{eqn: intermediate length scale}
    \fs:=\fn\ceil{K\ell_0/\fn},
\end{equation}
where $K>0$ is a constant to be chosen later.
Define the events $\Sigma_1\subset\Omega$ and $\Sigma_2\subset\Omega\times\Omega$ as follows:
\begin{enumerate}
    \item $\sigma\in\Sigma_1$ if and only if every dimer in $\sigma$ that intersects $\rect{1\times 3\fs}{(-1/2,-\fs-1/2)}$ is vertical.
    \item $(\sigma,\sigma')\in\Sigma_2$ if and only if the rectangles $\rect{1\times \fs}{(-1/2,-\fs-1/2)}$ and $\rect{1\times \fs}{(-1/2,\fs-1/2)}$ each contain a \emph{coincident} vacancy or vertical dimer in $\sigma$ and $\sigma'$, as illustrated in Figure~\ref{fig:sealing patterns}.
\end{enumerate}
By an abuse of notation, we write 
\begin{equation}
    \Sigma_1(\sigma):=\set{(\sigma,\sigma')\in\Omega\times\Omega\mid\sigma\in\Sigma_1}
    \text{ and }\Sigma_1(\sigma'):=\set{(\sigma,\sigma')\in\Omega\times\Omega\mid\sigma'\in\Sigma_1}.
\end{equation}
We say that the rectangle $\mathrm{R}_{1\times \fs}$ is \textbf{sealed} in $(\sigma,\sigma')$ if the event $\Sigma:=\Sigma_1(\sigma)\cap\Sigma_1(\sigma')\cap\Sigma_2$ holds.
Write $\Pi:=\xre{1\times \fs}{\es{\Sigma}}$.

\begin{figure}
    \centering
    \begin{subfigure}[t]{0.3\columnwidth}
        \centering
        \includegraphics{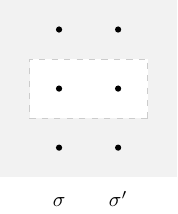}
        \caption{}
    \end{subfigure}
    ~
    \begin{subfigure}[t]{0.3\columnwidth}
        \centering
        \includegraphics{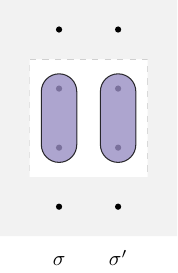}
        \caption{}
    \end{subfigure}
    \caption{Provided that a $1\times \fs$ rectangle intersects only vertical dimers in both $\sigma$ and $\sigma'$, no disagreement $\ddag$-path in $\Delta_{\sigma,\sigma'}$ can cross a coincident vacancy or vertical dimer in that rectangle.}
    \label{fig:sealing patterns}
\end{figure}

\begin{lemma}
    \label{lem:sealed rectangles - no horizontal dimers}
    There exists a constant $\beta_0>0$ such that, for all $\beta\ge\beta_0$ and $\Z^2$-ergodic Gibbs measures $\mu$ with $\Phase{\mu}=\ver$, with $C_\fn>0$ as in Theorem~\ref{thm: microscopic rarity of horizontal dimers},
    \begin{equation}
        \peierls{\mu}(\xre{\fn\times 3\fs}{\es{\Sigma_1}})
        \le 2\fn\sqrt[2\fn]{\frac{3C_\fn}{\fn} e^{-3\beta a(1-1/\fn)}\fs}
        .
    \end{equation}
\end{lemma}

\begin{proof}
    Recall the edge $\fe$ from the context of~\eqref{eqn: weak vertical order}, and let $\varphi$ denote the translation of $\R^2$ that maps $\fe$ to the edge $\set{(-1,-\fs),(0,-\fs)}$.
    Then,
    \begin{equation}
        \label{eqn: rewriting of Sigma_1}
        \Sigma_1=\bigcap_{i=0}^{1} \bigcap_{j=0}^{\fn-1} \eta_{(i,j)} \bigcap_{k=0}^{3\fs/\fn-1}\eta_{(0,k\fn)}\varphi F_\ver
        .
    \end{equation}
    By Theorem~\ref{thm: microscopic rarity of horizontal dimers} and Item~\ref{itm:strongly percolating grid events - subgrid over grid} of Lemma~\ref{lem:strongly percolating grid events}, we have that 
    \begin{equation}
        \peierls{\mu}\parentheses[\Bigg]{\xre{\fn\times 3\fs}{\es{\bigcap\nolimits_{k=0}^{3\fs/\fn-1}\eta_{(0,k\fn)}\varphi F_\ver}}}
        \le \frac{3C_\fn}{\fn}  e^{-3\beta a(1-1/\fn)}\fs.
    \end{equation}
    Thus, by Lemma~\ref{lem:intersection of strongly percolating sets}, 
    \begin{equation}
        \peierls{\mu}\parentheses{\xre{\fn\times 3\fs}{\es{\Sigma_1}}}
        \le 2\fn\sqrt[2\fn]{\frac{3C_\fn}{\fn} e^{-3\beta a(1-1/\fn)}\fs},
    \end{equation}
    as required.
\end{proof}

\begin{lemma}
    \label{lem:sealed rectangles - sealing patterns}
    For all $K>0$, there exist constants $\beta_K,C_K>0$ such that, for all $\beta\ge\beta_K$ and periodic Gibbs measures $\mu,\mu'$, with $\fs$ given by~\eqref{eqn: intermediate length scale},
    \begin{equation}
        \begin{multlined}
            \peierls{\mu\otimes\mu'}(\xre{1\times 3\fs}{\es{\Sigma_2\cup(\Sigma_1(\sigma)\cap\Sigma_1(\sigma'))^c}})
            \le 4e^{-\fs/(8\ell_0)}+C_K\brackets*{e^{-\frac{1}{2}\beta(\chempot+a)}+e^{-\frac{1}{2}\beta(\chempot+3a)}}.
        \end{multlined}
    \end{equation}
\end{lemma}

\begin{proof}
    Denote $S=\rect{1\times3\fs}{(-1/2,-\fs-1/2)}$.
    By Lemma~\ref{lem:strongly percolating grid events - chessboard bound for product measures},
    \begin{equation}
        \peierls{\mu\otimes\mu'}(\xre{1\times 3\fs}{\es{\Sigma_2\cup(\Sigma_1(\sigma)\cap\Sigma_1(\sigma'))^c}})\le\norm{(\Sigma_1(\sigma)\cap\Sigma_1(\sigma'))\setminus\Sigma_2}_{S}.
    \end{equation}
    With the goal of applying a union bound to the RHS, we fix a $1\times \fs$ rectangle 
    \begin{equation}
        R\in\set{\rect{1\times \fs}{(-1/2,-\fs-1/2)},\rect{1\times \fs}{(-1/2,\fs-1/2)}}
    \end{equation} 
    and study the event $E\subset\Sigma_1(\sigma)\cap\Sigma_1(\sigma')$ that $R$ does not contain any sealing pattern.
    We will write $E\subset E_1\cup E_{2}\cup E_2'\cup E_3$ and then estimate $\norm{E}_{R}$ by a separate union bound; in the definition of the events below, we will always assume implicitly that $R$ intersects only vertical dimers in both $\sigma$ and $\sigma'$.

    \begin{figure}
        \centering
        \begin{subfigure}[t]{0.3\columnwidth}
            \centering
            \includegraphics{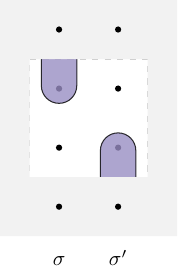}
            \caption{}
        \end{subfigure}
        ~
        \begin{subfigure}[t]{0.3\columnwidth}
            \centering
            \includegraphics{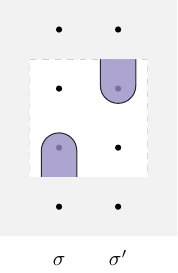}
            \caption{}
        \end{subfigure}
        \caption{Provided that $\fs$ is even and a $1\times\fs$ rectangle (i) intersects only vertical dimers in both $\sigma$ and $\sigma'$, (ii) is not fully packed with dimers in either $\sigma$ or $\sigma'$, and (iii) does not contain two adjacent vacancies in either $\sigma$ or $\sigma'$, the only way to rule out the existence of coincident vacancies and vertical dimers in the rectangle is to have a vacancy of $\sigma$ adjacent to a vacancy of $\sigma'$.}
        \label{fig:leaky patterns}
    \end{figure}

    Define $E_1$ as the event that $R$ is fully packed (by vertical dimers) in at least one of $\sigma$ and $\sigma'$.
    On $E_1^c$, both $\sigma$ and $\sigma'$ contain a vacancy in $R$.
    Define $E_2$ as the event that $\sigma$ has two adjacent vacancies in $R$, and define $E_2'$ analogously with respect to $\sigma'$.
    On $(E_1\cup E_2\cup E_2')^c$, both $\sigma$ and $\sigma'$ have only isolated vacancies in $R$.
    Since $\fs$ is even, by a case analysis, to rule out the existence of coincident vacancies and vertical dimers, there must exist a vacancy of $\sigma$ in $R$ that is adjacent to a vacancy of $\sigma'$ in $R$, as illustrated in Figure~\ref{fig:leaky patterns}.
    Define $E_3$ as the event that such a pair of adjacent vacancies exists in $R$; in particular, neither vacancy is required to be isolated.

    \begin{claim}
        For all $K>0$, there exist constants $\beta_K,C_K>0$ such that, for all $\beta\ge\beta_K$ and with $\fs$ given by~\eqref{eqn: intermediate length scale},
        \begin{enumerate}
            \item \label{itm:E1 chessboard seminorm} $\norm{E_1}_{R}\le 2e^{-\fs/(8\ell_0)}$,
            \item \label{itm:E2 chessboard seminorm} $\max\set{\norm{E_2}_{R},\norm{E_2'}_{R}} \le C_K e^{-\frac{1}{2}\beta(\chempot+a)}$,
            \item \label{itm:E3 chessboard seminorm} $\norm{E_3}_{R}\le C_K e^{-\frac{1}{2}\beta(\chempot+3a)}$.
        \end{enumerate}
    \end{claim}

    \begin{proof}
        Item~\ref{itm:E1 chessboard seminorm} follows from a union bound making use of Proposition~\ref{prop:partition function periodic boundary} and the fact that there exist exactly $2^{n!}$ configurations fully packed with vertical dimers in an $n!\times n!$ rectangle with periodic boundary conditions, all having weight $1$: 
        \begin{equation}
            \norm{E_1}_{R}
            \le 2 \limsup_{n\to\infty}\left[\frac{2^{n!}}{(1+\frac{1}{4}\ell_0^{-1})^{(n!)^2}}\right]^{\frac{\fs}{(n!)^2}}
            =2(1+\ell_0^{-1}/4)^{-\fs}
            \le 2e^{-\fs/(8\ell_0)}.
        \end{equation}
        
        For Item~\ref{itm:E2 chessboard seminorm}, define, for $1\le i\le \fs-1$, the event $E_{2,i}$ that $\sigma$ has vacancies at the $i$-th and $(i+1)$-th sites of $R$.
        Then,
        \begin{equation}
            \norm{E_{2,i}}_{R}
            \le\limsup_{n\to\infty}\left[(\bra{0}T\ket{0})^2\bra{0}T^{2i-1}\ket{0}\bra{0}T^{2\fs-2i-1}\ket{0}\right]^{\frac{n!}{2\fs}\frac{\fs}{n!}}.
        \end{equation}
        By~\eqref{eqn: transfer matrix} and Corollary~\ref{cor: partition function vacant boundary}, and using the AM-GM inequality, we have that
        \begin{equation}
            \norm{E_{2,i}}_{R}
            \le \frac{1}{2}C_K e^{-\frac{1}{2}\beta(\chempot+a)}\left[2e^{-\frac{1}{2}\beta(\chempot+3a)}+e^{-\frac{1}{2}\beta(\chempot+a)(2i-1)}+e^{-\frac{1}{2}\beta(\chempot+a)(2\fs-2i-1)}\right]
        \end{equation}
        for some constant $C_K>0$ depending only on $K$.
        By a union bound over $i$, we get that
        \begin{equation}
            \norm{E_2}_R
            \le \sum_{i=1}^{\fs-1}\norm{E_{2,i}}_R
            \le C_K e^{-\frac{1}{2}\beta(\chempot+a)}
            \brackets[\bigg]{(\fs-1) e^{-\frac{1}{2}\beta(\chempot+3a)}+\sum_{i=1}^{\fs-1} e^{-\frac{1}{2}\beta(\chempot+a)(2i-1)}},
        \end{equation}
        which implies the desired bound since $\fs=\order{K\ell_0}$.
        The same argument applies to $\norm{E_2'}_R$.

        The proof of Item~\ref{itm:E3 chessboard seminorm} is very similar.
        For $1\le i\le \fs-1$, define the event $E_{3,i+}$ that $\sigma$ has a vacancy at the $i$-th site of $R$ and $\sigma'$ has a vacancy at the $(i+1)$-th site of $R$, and similarly the event $E_{3,i-}$ with the roles of $\sigma$ and $\sigma'$ reversed.
        Then,
        \begin{equation}
            \norm{E_{3,i+}}_{R}
            \le\limsup_{n\to\infty}\left[\bra{0}T^{2i-1}\ket{0}\bra{0}T^{2\fs-2i+1}\ket{0}\bra{0}T^{2i+1}\ket{0}\bra{0}T^{2\fs-2i-1}\ket{0}\right]^{\frac{n!}{2\fs}\frac{\fs}{n!}}.
        \end{equation}
        By Corollary~\ref{cor: partition function vacant boundary}, we have that
        \begin{equation}
            \norm{E_{3,i+}}_{R}
            \le C_K^2 
            \left[e^{-\frac{1}{2}\beta(\chempot+3a)}+e^{-\frac{1}{2}\beta(\chempot+a)(2i-1)}\right]
            \left[e^{-\frac{1}{2}\beta(\chempot+3a)}+e^{-\frac{1}{2}\beta(\chempot+a)(2\fs-2i-1)}\right]
        \end{equation}
        for the same constant $C_K>0$ as above; an analogous bound holds for $\norm{E_{3,i-}}_{R}$.
        By a union bound,
        \begin{equation}
            \begin{multlined}
                \norm{E_3}_R
                \le \sum_{i=1}^{\fs-1}\parentheses[\big]{\norm{E_{3,i+}}_{R}+\norm{E_{3,i-}}_{R}}
                \\
                \le 2C_K^2\sum_{i=1}^{\fs-1}\left[e^{-\frac{1}{2}\beta(\chempot+3a)}+e^{-\frac{1}{2}\beta(\chempot+a)(2i-1)}\right]
                \left[e^{-\frac{1}{2}\beta(\chempot+3a)}+e^{-\frac{1}{2}\beta(\chempot+a)(2\fs-2i-1)}\right],
            \end{multlined}
        \end{equation}
        which implies the desired bound.
    \end{proof}

    The proof of the lemma is complete after applying Lemma~\ref{lem:infinite-volume recursive chessboard estimate for product measures} (noting that the corners of $S$ are on $G^R$) and a union bound over the events $E_1,E_2,E_2',E_3$ and the two choices of $R$.
\end{proof}

We are now ready to prove the strong percolation of sealed rectangles in the vertical phase.

\begin{proposition}
    \label{prop:sealed rectangles strongly percolate}
    For every $\eps>0$, one may choose $K\ge 1$ and then $\beta_0>0$ such that, for all $\beta\ge\beta_0$ and $\Z^2$-ergodic Gibbs measures $\mu,\mu'$ with $\Phase{\mu}=\Phase{\mu'}=\ver$, and with $\fs$ given by~\eqref{eqn: intermediate length scale}, 
    \begin{equation}
        \peierls{\mu\otimes\mu'}(\xre{1\times\fs}{\es{\Sigma}})\le \eps.
    \end{equation}
\end{proposition}

\begin{proof}
    Let $\eps>0$.
    Write
    \begin{equation}
        \Sigma=\Sigma_1(\sigma)\cap\Sigma_1(\sigma')\cap\left[\Sigma_2\cup(\Sigma_1(\sigma)\cap\Sigma_1(\sigma'))^c\right]
        .
    \end{equation}
    By Lemma~\ref{lem:intersection of strongly percolating sets},
    \begin{equation}
        \begin{multlined}
            \peierls{\mu\otimes\mu'}(\xre{1\times \fs}{\es{\Sigma}})
            \le3\max\left\{
            \peierls{\mu\otimes\mu'}(\xre{1\times \fs}{\es{\Sigma_1(\sigma)}}),
            \peierls{\mu\otimes\mu'}(\xre{1\times \fs}{\es{\Sigma_1(\sigma')}}),
            \right.
            \\
            \left.
            \peierls{\mu\otimes\mu'}(\xre{1\times \fs}{\es{\Sigma_2\cup(\Sigma_1(\sigma)\cap\Sigma_1(\sigma'))^c}})
            \right\}^{1/3}.
        \end{multlined}
    \end{equation}

    Let $\delta>0$.
    By taking $K\ge 1$ and then $\beta_0>0$ sufficiently large, we may ensure that the bound given by Lemma~\ref{lem:sealed rectangles - sealing patterns} is no greater than $\delta$.
    Recall the choice of $\fn$ in~\eqref{eqn: microscopic length scales} and note that
    \begin{equation}
        3\parentheses[\bigg]{1-\frac{1}{\fn}}a-\frac{1}{2}(\chempot+3a)
        >3\parentheses[\bigg]{1-\frac{3a-\chempot}{24a}}a-\frac{1}{2}(\chempot+3a)
        =\frac{3}{8}(3a-\chempot)
        >0
    \end{equation}
    by the assumption $3a>\chempot$.
    Thus, by further increasing $\beta_0$ if necessary, we may as well ensure that the bound given by Lemma~\ref{lem:sealed rectangles - no horizontal dimers} is no greater than $\delta$.
    It is clear that
    \begin{equation}
        \peierls{\mu\otimes\mu'}(\xre{1\times \fs}{\es{\Sigma_1(\sigma)}})
        \le \peierls{\mu}(\xre{1\times \fs}{\es{\Sigma_1}}),
        \quad
        \peierls{\mu\otimes\mu'}(\xre{1\times \fs}{\es{\Sigma_1(\sigma')}})
        \le \peierls{\mu'}(\xre{1\times \fs}{\es{\Sigma_1}}).
    \end{equation}
    Hence, by Item~\ref{itm:strongly percolating grid events - grid over subgrid} of Lemma~\ref{lem:strongly percolating grid events}, we get that
    \begin{equation}
        \begin{multlined}
            \peierls{\mu\otimes\mu'}(\xre{1\times \fs}{\es{\Sigma}})
            \le3\max\set{
            \sqrt[3\fn]{3\fn\sqrt[3\fn]{\delta}},
            \sqrt[3]{3\sqrt[3]{\delta}}}^{1/3}.
        \end{multlined}
    \end{equation}
    The proof is complete by taking $\delta$ sufficiently small.
\end{proof}

\subsubsection{Bounding the disagreement components}

The following deterministic statement clarifies the sense in which sealed rectangles rule out the presence of large disagreement $\ddag$-components nearby.

\begin{proposition}
    \label{prop:bounded disagreement components in sealed rectangles}
    Let $(\sigma,\sigma')\in\Omega\times\Omega$. 
    Suppose that $S=\rect{1\times\fs}{(x_0-1/2,\fs y_0-1/2)}$ is sealed in $(\sigma,\sigma')$ for some $(x_0,y_0)\in\Z^2$.
    Let $e=(x_e,y_e)\in\V\cap S$.
    If $e\in\Delta_{\sigma,\sigma'}$, then $e$ is necessarily vertical, and the $\ddag$-component of $e$ in $\Delta_{\sigma,\sigma'}$ is contained in $\set{x_e}\times(\fs(y_0-1)-1/2,\fs(y_0+2)-1/2)$.
\end{proposition}

\begin{proof}
    Since $\sigma,\sigma'\in\eta_{(x_0,\fs y_0)}\Sigma_1$, $e$ is necessarily vertical, and the $\ddag$-component of $e$ in $\Delta_{\sigma,\sigma'}$ is a vertical $\ddag$-path.
    This $\ddag$-path terminates at the first instance of a coincident vacancy or vertical dimer in $\sigma$ and $\sigma'$, but such a pattern exists in both $\rect{1\times\fs}{(x_0-1/2,\fs(y_0-1)-1/2)}$ and $\rect{1\times \fs}{(x_0-1/2,\fs(y_0+1)-1/2)}$ since $(\sigma,\sigma')\in\eta_{(x_0,\fs y_0)}\Sigma_2$.
\end{proof}

\subsubsection{Proof of Proposition~\ref{prop:disagreement percolation decay}}

With $\eps=\eps_0/2$, choose $K\ge 1$ and $\beta_0>0$ so that the conclusion of Proposition~\ref{prop:sealed rectangles strongly percolate} holds.
Let $\alpha_1(u,v)$ be defined as in~\eqref{eqn: Heilmann-Lieb decay}, where $C\ge 1$ and $c>0$ will be determined at the end of the proof.

Recall the identification $\V\equiv[(1/2+\Z)\times\Z]\cup[\Z\times(1/2+\Z)]$.
Define $f:\V\to\Z^2$ by $f(x,y):=(\floor{x},\floor{\frac{y}{\fs}})$.
Let $\sigma,\sigma'$ be as in the statement of the proposition.
We will rely on the following claim that long disagreement paths in $\Delta_{\sigma,\sigma'}$ induce long sequences of non-sealed rectangles.

\begin{claim}
    \label{clm:disagreement paths induce box paths}
    Let $u,v\in\V$.
    Suppose that $u-v\not\in\set{0}\times[-4\fs,4\fs]$ and that $u,v$ are connected by a $\ddag$-path in $\Delta_{\sigma,\sigma'}$.
    Then, $f(u),f(v)$ are connected by a $\boxtimes$-path in $\Z^2\setminus\Pi$.
\end{claim}

\begin{proof}
    Let $P$ be a $\ddag$-path in $\Delta_{\sigma,\sigma'}$ connecting $u$ and $v$.
    The map $f$ sends consecutive \textit{edges} in $P$ (which are elements of $\V$) to either the same vertex or $\boxtimes$-adjacent vertices in $\Z^2$, except when they are both horizontal and separated by one other horizontal edge.
    In the exceptional case, their images under $f$ differ by $(\pm 2,0)$, so there is a unique vertex of $\Z^2$ in between.
    Upon inserting these intermediate vertices into $f(P)$, we obtain a $\boxtimes$-walk $Q$ in $\Z^2$ connecting $f(u)$ to $f(v)$.
    It suffices to show that $Q\subset\Z^2\setminus\Pi$.

    First, note that, for each $z=(x_z,y_z)\in Q$, there exists $w\in P$ such that $w\in R_z:=\rect{1\times\fs}{(x_z-1/2,\fs y_z-1/2)}$.
    Indeed, if $z\in f(P)$, this follows directly from the definition of $f$.
    Otherwise, $z$ was inserted into $f(P)$, in between the images of two separated horizontal \textit{edges} $w_1,w_2\in P$, one of which necessarily belongs to $R_z$.

    By contradiction, suppose that $Q\cap\Pi\ne\emptyset$.
    Let $z=(x_z,y_z)\in Q\cap\Pi$, and choose $w=(x_w,y_w)\in P$ as in the previous paragraph.
    Then, $R_z$ is sealed.
    By Proposition~\ref{prop:bounded disagreement components in sealed rectangles}, $w$ is vertical, and the $\ddag$-component of $w$ in $\Delta_{\sigma,\sigma'}$, which includes $u$ and $v$, is contained in $\set{x_w}\times[y_w-2\fs,y_w+2\fs]$, contradicting the assumption that $u-v\not\in\set{0}\times[-4\fs,4\fs]$.
\end{proof}

Let $u=(x_u,y_u)\in A$.
By a union bound, it suffices to show that
\begin{equation}
    (\mu\otimes\mu')(\{u\}\text{ and }B\text{ are connected by a $\ddag$-path in }\Delta_{\sigma,\sigma'})\le \sup_{v\in B}\alpha_1(u,v),
\end{equation}
Suppose that $u$ is connected by a $\ddag$-path in $\Delta_{\sigma,\sigma'}$ to some $v=(x_v,y_v)\in B$.
If there exists such a $v$ with $u-v\in\set{0}\times[-4\fs,4\fs]$, then
\begin{equation}
    \alpha_1(u,v)\ge C\exp{-4c\ell_0^{-1}\fs}.
\end{equation}
Otherwise, by Claim~\ref{clm:disagreement paths induce box paths}, $f(u)$ is connected to $f(v)$ by a $\boxtimes$-path in $\Z^2\setminus\Pi$.
By Proposition~\ref{prop:sealed rectangles strongly percolate} and Lemma~\ref{lem:quantitative Peierls argument}, the probability of this event is bounded above by
\begin{equation}
    (\eps/\eps_0)^{1+\inf_{v\in B}\norm{f(v)-f(u)}_\infty}
    \le (\eps/\eps_0)^{\frac{1}{2}\inf_{v\in B}\left(\abs{x_u-x_v}+\frac{\abs{y_u-y_v}}{\fs}\right)}
    =\sup_{v\in B}(\eps/\eps_0)^{\frac{1}{2}\left(\abs{x_u-x_v}+\frac{\abs{y_u-y_v}}{\fs}\right)}.
\end{equation}
Choosing $C\ge 1,\,c>0$ so that 
\begin{equation}
    C\exp{-4c\ell_0^{-1}\fs}\ge 1,\quad(\eps/\eps_0)^{\frac{1}{2}}\le e^{-c},\quad\text{and}\quad(\eps/\eps_0)^{\frac{1}{2\fs}}\le e^{-c\ell_0^{-1}}
\end{equation}
completes the proof.

\subsection{Microscopic characterizations of nematic order}
\label{sec:microscopic characterization of nematic order}

In this subsection, we derive microscopic characterizations of nematic order in terms of the probability of observing a vertical dimer versus a horizontal one within the vertical phase, thus establishing Item~\ref{itm:main - dimer probabilities} of Theorem~\ref{thm:main}.

First, we estimate the probability of observing a horizontal dimer in the vertical phase.

\begin{theorem}
    \label{thm:horizontal dimers are unlikely}
    Suppose that $\beta\ge\beta_0$.
    If $e\in\V$ is horizontal, then
    \begin{equation}
        \mu_\ver(\sigma(e)=1)=\order{e^{-3\beta(1-1/\fn) a}}.
    \end{equation}
\end{theorem}

\begin{proof}
    Recall the edge $\fe$ from the context of~\eqref{eqn: weak vertical order}.
    Since $\mu_\ver$ is $\Z^2$-invariant, we have that
    \begin{equation}
        \mu_\ver(\sigma(e)=1)
        =\mu_\ver(\sigma(\fe)=1)
        \le \frac{\peierls{\mu_\ver}(\Theta_\ver)}{\eps_0}
    \end{equation}
    by Lemma~\ref{lem:quantitative Peierls argument}.
    Applying Theorem~\ref{thm: microscopic rarity of horizontal dimers} completes the proof.
\end{proof}

Second, we estimate the probability of observing a vertical dimer in the vertical phase.

\begin{theorem}
    Suppose that $\beta\ge\beta_0$.
    If $e\in\V$ is vertical, then
    \begin{equation}
        \mu_\ver(\sigma(e)=1)=\frac{1}{2}-\order{e^{-\frac{1}{2}\beta (\chempot+a)}}.
    \end{equation}
\end{theorem}

\begin{proof}
    Let $e=(x,y)\in\V$ be vertical.
    Consider the events
    \begin{align}
        E_1{}&:=\set{\sigma(x,y)+\sigma(x,y+1)=1}, \\
        E_2{}&:=\set{\sigma(x-1/2,y+1/2)+\sigma(x+1/2,y+1/2)=1}, \\
        E_3{}&:=\set{(x,y+1/2)\text{ is a vacancy of }\sigma}.
    \end{align}
    Note that these events partition $\Omega$.
    By Theorem~\ref{thm:horizontal dimers are unlikely}, $\mu_\ver(E_2)=\order{e^{-3\beta(1-1/\fn) a}}$.
    By Proposition~\ref{prop:infinite-volume chessboard estimate for product measures} and Item~\ref{itm:vacancy chessboard seminorm} of Corollary~\ref{cor:local defect chessboard seminorms}, $\mu_\ver(E_3)\le e^{-\frac{1}{2}\beta(\chempot+a)}$.
    Hence, 
    \begin{equation}
        \mu_\ver(E_1)=1-\mu_\ver(E_2)-\mu_\ver(E_3)=1-\order{e^{-\frac{1}{2}\beta (\chempot+a)}},
    \end{equation}
    having used that $3a>\chempot$ and $\fn\ge 100$ in the second equality.
    Since $\mu_\ver$ is $\Z^2$-invariant, $\mu_\ver(E_1)=2\mu_\ver(\sigma(e)=1)$, which implies the theorem.
\end{proof}

\section{On the condition $3a>\chempot$}
\label{sec:origin of extra condition}

Here, we reflect on the role of the condition $3a>\chempot$ in the proof of Theorem~\ref{thm:main} and identify some obstacles to extending the current proof to the remaining regime.

In this paper, we have imposed three conditions on the parameters $\chempot$ and $a$: $a>0$, $\chempot+a>0$, and $3a>\chempot$, which play distinct roles in the proof of Theorem~\ref{thm:main}.
The assumption $a>0$, built into the model, makes the alignment between adjacent, colinear dimers energetically favorable; it is well-known that the monomer-dimer model has no phase transition without this attractive interaction~\cite{heilmann1972theory,van1999absence}.
As the rewriting of the Hamiltonian in~\eqref{eqn:shifted Hamiltonian monomer-dimer} makes clear, the condition $\chempot+a>0$ suppresses vacancies at low temperatures, leading to densely packed configurations in which orientational order is manifest.
The role of the remaining condition $3a>\chempot$ is more subtle and is tied to the separation between two relevant low-temperature length scales, as we now discuss.

To understand its origin, it is instructive to consider the two \emph{simplest} defects in an otherwise fully packed and perfectly aligned dimer configuration: an isolated vacancy and an isolated misaligned dimer.
In view of the Hamiltonian in~\eqref{eqn:shifted Hamiltonian monomer-dimer}, the former produces two broken links upon its creation and therefore raises the energy by $(\chempot+3a)/2$, whereas the latter produces six broken links and therefore raises the energy by $3a$.
A direct comparison shows that the former is energetically cheaper if and only if $3a>\chempot$, with a crossover on the half-line $3a=\chempot$.

This energy disparity gives rise to a separation-of-scales mechanism exploited in the construction and proof of the strong percolation of sealed rectangles in Section~\ref{sec:sealed rectangles}.
In the vertical phase, a $1\times\fs$ rectangle must be sufficiently tall for two typical, independently sampled configurations to have a coincident vacancy or vertical dimer near each of its ends.
Lemma~\ref{lem:sealed rectangles - sealing patterns} shows that the dominant obstruction to this coincidence comes from having a fully packed one-dimensional configuration in the rectangle, whose contribution is bounded by a term of the form $e^{-c\ell_0^{-1}\fs}$, so $\fs$ must be taken large compared to $\ell_0$ to suppress this obstruction.
At the same time, Theorem~\ref{thm: microscopic rarity of horizontal dimers} and Lemma~\ref{lem:sealed rectangles - no horizontal dimers} show that the incidence of misaligned dimers in such a rectangle is controlled via the quantity $e^{-3\beta a(1-1/\fn)}\fs$, which arises from promoting the microscopic rarity of misaligned dimers to the mesoscopic scale.
Consequently, after taking $\fs$ to be of order $\ell_0$, we need $3a(1-1/\fn)>(\chempot+3a)/2$ to suppress horizontal dimers in the rectangle, which in turn leads to the second, $(\chempot,a)$-dependent term in the definition of $\fn$ in~\eqref{eqn: microscopic length scales}.

In summary, if we neglect the loss of $1/\fn$ from localizing the defects in the proof of Theorem~\ref{thm: microscopic rarity of horizontal dimers}, the underlying requirement of the argument can be formulated as the existence of a length scale $\fs$ such that
\begin{equation}
    \ell_0\ll \fs\ll e^{3\beta a},
\end{equation}
which exists precisely when $3a>\chempot$.
This picture also explains why the extra condition is not an artifact of the particular definition of sealed rectangles used in the proof. 
Indeed, any argument that views typical configurations in the vertical phase as essentially one-dimensional over distances comparable to the one-dimensional correlation length $\ell_0$ implicitly requires the incidence of horizontal dimers to remain negligible on that scale.

When $3a<\chempot$, the picture is reversed. 
In the vertical phase, a rectangle of height comparable to $\ell_0$ typically witnesses many horizontal dimers before finding a vacancy, whereas a rectangle of height comparable to $e^{3\beta a}$ is too short for the sealing patterns considered in Section~\ref{sec:sealed rectangles} to appear with high probability.
To make matters worse, non-coincident horizontal dimers allow disagreement paths to propagate across columns, destroying the deterministic single-column confinement established in Proposition~\ref{prop:bounded disagreement components in sealed rectangles}.
Therefore, a proof of liquid-crystalline order in this regime would require a genuinely two-dimensional analysis, starting with a suitable characterization of the reference systems and some notion of sealing that tolerates typical defects in these systems.

The borderline case $3a=\chempot$ is also not covered by the present argument.
Here, the separation-of-scales mechanism breaks down completely, and vacancies and misaligned dimers arise at comparable rates.
It is plausible that progress in the regime $3a<\chempot$ would shed light on this case, since the non-negligible incidence of misaligned dimers likely also necessitates a two-dimensional analysis.

In conclusion, it remains open whether nematic order persists throughout the entire regime $\chempot+a>0$ as Heilmann--Lieb conjectured~\cite{heilmann1979lattice}.
In particular, when $3a<\chempot$, it is conceivable either that the two types of elementary defects identified earlier destroy all translational order, or that misaligned dimers, through parity constraints imposed by the aligned background, develop a parity preference transverse to the direction of alignment and split each orientational phase into two, resulting in columnar order\footnote{We thank an anonymous referee for pointing out this possibility.} as is the case for the hard-square model studied by Hadas--Peled~\cite{hadas2025columnar}.

\subsection*{Acknowledgments}

The author thanks Profs. Ian Jauslin and Ron Peled for suggesting the problem and Ron Peled in particular for many enlightening discussions and for pointing out that the chessboard estimate applies to independent couplings of reflection-positive measures.
The author also thanks Minhao Bai and Anupam Nayak for helpful discussions.
The author is grateful to an anonymous referee for their careful reading of the manuscript and invaluable comments and insights that led to a significantly simplified proof strategy.
Most of this work was completed while the author was a PhD student at Rutgers University, where he was supported by an SAS Fellowship and in part through a summer coadjutant appointment after his graduation, made possible by a gift from Prof. Joel Lebowitz.

\subsection*{Data availability}
No datasets were generated or analyzed during the current study.

\section*{Declarations}

\subsection*{Conflict of interest}
The author has no relevant financial or non-financial interests to disclose.

\bibliographystyle{plain}
\bibliography{bibliography}

\end{document}